\documentclass{statsoc}

\usepackage[a4paper]{geometry}          
\usepackage[textwidth=8em,textsize=small]{todonotes} 
\usepackage{amsmath}
\usepackage{amssymb}
\usepackage{graphicx,psfrag,epsf}
\usepackage{enumerate}
\usepackage{natbib}
\usepackage{url}                
\usepackage{xr}
\usepackage{times}
\usepackage[plain,noend]{algorithm2e}
\usepackage{enumitem}
\usepackage[utf8]{inputenc}
\usepackage{color,soul}
\allowdisplaybreaks         

\newcommand{\argmax}{\operatornamewithlimits{argmax}}

\newcommand{\Nn}[1]{N_{n}(#1)}
\newcommand{\Dn}[1]{D_{n}(#1)}

\newcommand{\gn}{\hat{g}_{n}}
\newcommand{\g}{g_0}
\newcommand{\G}{\mathbb{G}}

\newcommand{\intinfinf}{\int_{-\infty}^{\infty}}

\newcommand{\intainf}[1]{\int_{a}^{\infty}}

\newcommand{\SUM}[3]{\sum_{#1=#2}^{#3}}
\newcommand{\nSUM}[4]{\sum^{#3}_{\substack{#1=#2 \\ (\neq #4)}}}
\numberwithin{equation}{section}
\newtheorem{theorem}{Theorem}[section]

\newtheorem{lemma}[theorem]{Lemma}

\title[Curve Registration]{Feature Sensitive Curve Registration by Kernel Matching}

\author{Dibyendu Bhaumik}
\address{Department of Statistics and Information Management, Reserve Bank of India,
        Mumbai,
        India.}
\email{dbhaumik@rbi.org.in}
\author{Radhendushka Srivastava}
\address{Department of Mathematics, Indian Institute of Technology,
        Mumbai,
        India.}
\email{rsrivastava22@gmail.com}
\author[Bhaumik {\it et al.}]{Debasis Sengupta\thanks{
    Address for correspondence: Debasis Sengupta, Applied Statistics Unit, Indian Statistical Institute, 203, Barrackpore Trunk Road, Kolkata 700108, West Bengal, India, E-mail: sdebasis@isical.ac.in}}
\address{Applied Statistics Unit, Indian Statistical Institute,
        Kolkata,
        India.}
\email{sdebasis@isical.ac.in}

\begin{document}
    \begin{abstract}
            In this paper, we argue that the problem of registering two sets of functional data, where the underlying mean function has sharp features, is not properly addressed by methods designed to align a bunch of growth curves data. We provide a new method, which is able to pool local information without smoothing and to match sharp landmarks without manual identification. This method, which we refer to as kernel-matched registration, is based on maximizing a kernel-based measure of alignment. We prove that the proposed method is consistent under fairly general conditions. Simulation results show superiority of the performance of the proposed method over two existing methods. The proposed method is illustrated through the analysis of three sets of paleoclimatic data.
    \end{abstract}


    \section{Introduction}\label{sec:int}

        Consider functional data arising from observations recorded at a sequence of time points. The task of aligning multiple but similar sets of functional data by possibly nonlinear adjustment to their time scales is often referred to as `registration'. Multi-dimensional versions of the problem of registration are important in image and video processing, where multiple dimensions are involved. In the one dimensional case, the focus of research has been in the area of growth curves. For longitudinal growth data, often viewed as a common pattern expressed differently through different individuals with their diverse scales of evolution, the need for registration arises from the quest for the common pattern. This is in contrast with the field of image-processing, where registration is used mostly for comparing images and modeling changes in them.

        In growth data applications, $k$ sets of functional data are often postulated as arising from the signal-plus-noise model
        \begin{equation}
            y_i(t) = \mu_i(g_i(t))+\epsilon_i(t),\quad i=1,\ldots,k,\label{eq:general_model1}
        \end{equation}
        where $g_1,\ldots,g_k$ are strictly increasing time-warping functions for different individuals, \linebreak$\mu_1,\ldots,\mu_k$ are smooth functions having common features, such as extrema or points of inflection, at identical points (across all individuals), and $\epsilon_i$'s are additive errors. The functions $\mu_1,\ldots,\mu_k$ are assumed to be variations of a common underlying function $\mu$. Once the warping functions are estimated, they are used to bring the data to a common time-scale, so that the `central' function $\mu$ can be estimated from the pooled data. Functional data on many individuals is expected to produce a good estimator of~$\mu$.

        A special case of~\eqref{eq:general_model1} with linear $g_i$'s and $\mu_i$'s that are linear variations of $\mu$, called a shape invariant model (SIM), had been studied by several researchers (see \citet{Lawton_Sylvestre_Maggio_1972}, \citet{Stutzle_et_al_1980}, \citet{Kneip_Gasser_1988}, \citet{Kneip_Engel_1995}). \citet{Brumback_Lindstrom_2004} worked on a more general model with  $g_1$,\ldots, $g_k$ and $\mu$ as splines, while \citet{Gervini_Gasser_2004} took $g_1^{-1}$,\ldots, $g_k^{-1}$ as splines.

        Borrowing the idea of {\it dynamic time warping} from the engineering literature, \citet{Wang_Gasser_1997} proposed pre-smoothing of the data to obtain individual-specific functions and their registration with respect to an arbitrary reference curve by minimization of a complex cost function that penalizes departure not only in the normed function but also in its normed derivative. This cost function gives special attention to growth curve modeling. Eventually the average of the estimated warping functions is chosen as the central time scale, and all the curves are re-registered with respect to it. \citet{Wang_Gasser_1999} modified this method by replacing the fixed reference curve with a dynamically updated `central' curve at every stage of iteration.

        \citet{Kneip_Gasser_1992} had considered general time-warping functions and formalized the intuitive approach of matching common features called `landmarks' for registration. Once the landmarks are matched, the rest of the warping function is obtained through interpolation. Other methods of this category, known as {\it marker registration}, are described in~\citet{Bookstein_1991}. \citet{Bigot_2006} proposed a method of automatically identifying landmarks and subsequent registration by minimizing a cost function that rewards proximity of landmarks but also rewards smoothness of the warping function.

        \citet{Ramsay_Li_1998} considered the following special case of~\eqref{eq:general_model1}
        \begin{equation}
            y_i(t) = \mu(g_i(t))+\epsilon_i(t),\quad i=1,\ldots,k.\label{eq:other_model1}
        \end{equation}
        They used a cost function with a different penalty that focuses on the curvature of the candidate warping function, constrained to be monotonically non-decreasing, and called this method continuous monotone registration. This method also starts with pre-smoothed data sets. At every stage of the iterative optimization process, the estimated $\mu$ is regarded as the sample average of the currently aligned versions of these smoothers.

        Many other methods have been proposed for the estimation of $g_i$ under the model~\eqref{eq:other_model1}. \citet{Kneip_et_al_2000} estimated the $g_i$'s iteratively by local non-linear regression. \citet{Gervini_Gasser_2005} assumed that $g_i$'s are parametric variations of a common positive-valued and increasing function, modeled through splines, and proposed nonparametric maximum likelihood estimation of this common function and~$\mu$. \citet{Liu_Muller_2004} suggested simultaneous registration of all the individual time-scales with respect to a reference scale, for which they recommended the cumulative proportion of the total area under a pre-smoothed version of $|y_i|$ till a particular time. \citet{James_2007} proposed a method based on matching of `functional moments', intended to capture landmarks or local features. \citet{Tang_Muller_2008} sought to synchronize smoothers through individual data sets in a pair-wise manner by minimizing an integrated squared distance cost function with penalty on departure of the warping function from the identity map. Collation of these warping functions leads to a common time scale for final registration. \citet{Kneip_Ramsay_2008} proposed registration under the principle that proper alignment of smoothers through individual data sets would be better approximated by functional principle component analysis.

        While many of the foregoing methods have applicability beyond growth curves, the statistical literature on registration has evolved with the growth curve example lying firmly at the centre of attention. Growth curve examples have been given even in papers that developed asymptotic results with number of longitudinal data points (rather than the number of replicates) going to infinity \citep{Kneip_Engel_1995, Wang_Gasser_1999, Gervini_Gasser_2004}. This singular focus may not have served other applications well, as we shall see.

        In the case of growth curves, the number of individuals is generally much more than the number of observations per individual. Another aspect of this problem is that the underlying function is known to be smooth. However, the need for registration may also arise in situations where the underlying function is not very smooth and the number of observations per data set is much more than the number of data sets to be time-aligned or registered. There may even be only two data sets for alignment.  Consider, for example, the atmospheric concentration of carbon dioxide for the past 415,000 years constructed from two ice cores extracted from two different locations in Antarctica namely Lake Vostok and EPICA Dome C \citep{Petit_et_al_1999, Luthi_et_al_2008}, plotted in Figure~\ref{fig:co2_epica_vostok}. The sharp ups and downs observed in the two time series at nearly identical points of time indicate that a common function having these features underlies the two sets of measurements. One cannot expect many replicates of ice core data to come about -- not only for economic reasons, but also because of provisions of international treaties that prohibit intrusion in an ecologically sensitive area \citep{Watts_1992}. Therefore, it is all the more important that information on the two sets of measurements are collated. However, the data are misaligned, possibly because of estimation error in the recorded `time'.
        \begin{figure}[h!]
            \begin{center}
                \includegraphics{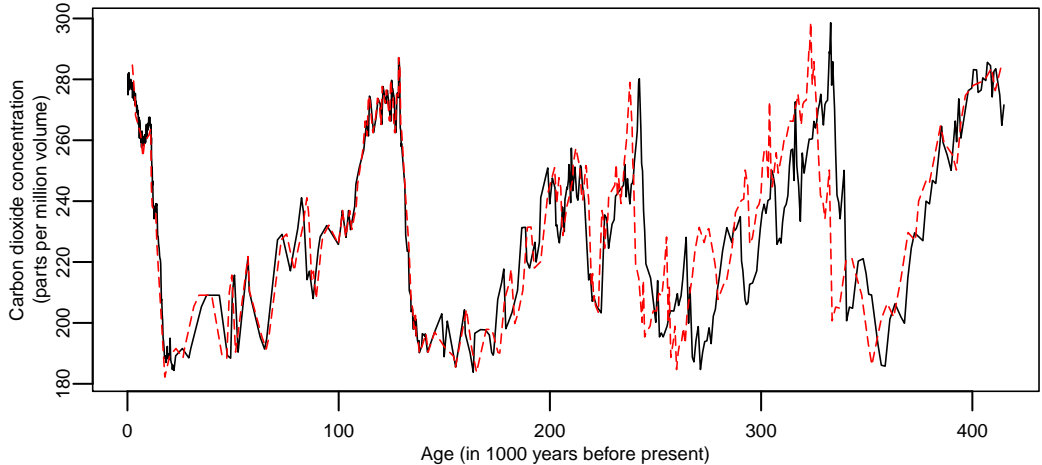}
                \caption{\label{fig:co2_epica_vostok}Atmospheric concentration of carbon dioxide derived from ice-cores at EPICA Dome C (solid line) and Lake Vostok (dashes) of Antarctica}
            \end{center}
        \end{figure}

        The sharp and numerous peaks and valleys in the above data series, often identified as `landmarks', have no parallel in longitudinal growth data. These features present an opportunity, as it may be possible to use them for the purpose of registration.

        A simple model for registering one data set with respect to another is
        \begin{equation}
            \begin{array}{rcl}
                y_1(t)&=&m(t) + \epsilon_1(t),\\
                y_2(t)&=&m(g_0(t)) + \epsilon_2(t),
            \end{array}\label{eq:the_model1}
        \end{equation}
        where $m$ is the underlying mean function, expressed in the time-scale of the first data set, and $g_0$ is the strictly increasing function that warps the time-scale of the second data set into that of the first.

        The model \eqref{eq:the_model1} may appear to be a simplified form of \eqref{eq:other_model1} for $k=2$, under the {\it asymmetric} constraint $g_1(t)=t$. In contrast, the more general models are often used for estimating the $g_i$'s subject to a {\it symmetric} constraint for identifiability. When there are only two sets of data ($k=2$), however, a neutral time-scale induced by a symmetric constraint appears unnecessary. In fact, \citet{Bhaumik_Sengupta_2017} have shown that for $k=2$, the model~\eqref{eq:other_model1} under a symmetric constraint may be even \textit{more} restrictive than the model \eqref{eq:the_model1}. Thus, it is necessary to pay special attention to the model \eqref{eq:the_model1} with possibly non-smooth $m$, which has so far gone under the radar of the models and methods designed to register growth curves.

        The key method that we present in this paper, which we call kernel-matched registration (KMR), combines the ideas of smoothing and matching of landmarks. Rather than aligning {\it smoothers through data sets} by minimizing the distance between them (as practised in dynamic time warping), we directly align the {\it data sets} by maximizing a measure of their alignment. This measure rewards matching of sharp features, which is the strong point of marker registration, and also pools local information, which is the strong point of methods based on smoothing.

        In \S~\ref{sec:mod_met}, we propose a new estimator of the time transformation function $g$, by maximizing a measure of alignment of two functional data sets. The maximization is done over an appropriate class of transformations. This measure of alignment is designed to capture the information contained in the entire set of data, including locations of sharp variation. We show that it possesses some desired characteristics. The method is automated as it does not require manual identification of landmarks. Identifiability of $g$ with respect to the chosen model, under appropriate conditions, is proved in \S~\ref{sec:con}. In \S~\ref{sec:con}, we have established the consistency of our estimator as the numbers of observations in the two data sets go to infinity. This result holds when the time transformation is chosen from a class of functions satisfying some general conditions. We have provided a method of estimating standard error of the proposed estimator in \S~\ref{sec:se_g_est}. We have reported in \S~\ref{sec:sim_per} the results of a simulation study to demonstrate the performance of the proposed estimator (chosen from a particular class), and the associated standard error. The method is illustrated in \S~\ref{sec:dat_ana} through the analysis of several real data sets. Some concluding remarks are provided in \S~\ref{sec:dis}.

    \section{Model and Methodology}\label{sec:mod_met}

        Let $\{(t_1, y_{11}), \ldots, (t_{n_1}, y_{1n_1})\}$ and $\{(s_1,y_{21}),\ldots,(s_{n_2},y_{2n_2})\}$ be two sets of functional data, arising from the model
        \begin{eqnarray}\label{eq:model}
            y_{1i} &=& m(t_i)+\epsilon_{1i}\quad(i=1,\ldots, {n_1})\notag\\
            y_{2j}&=&m(g_0(s_j))+\epsilon_{2j},\quad (j=1,\ldots, {n_2})
        \end{eqnarray}
        where $m$ is an underlying location function that is continuous, and $g_0$ is an unknown time transformation function, which is continuous and strictly increasing. The terms $\epsilon_1$ and $\epsilon_2$ represent additive random measurement errors, which have mean zero.

        For any given continuous and strictly increasing transformation function $g$, let us define the functional
        \begin{equation}\label{eq:Ln}
            L_n(g)=
            \frac
            {
            \displaystyle\frac{1}{n_1 n_2}
            \displaystyle\sum\limits_{i=1}^{n_1}
            \displaystyle\sum\limits_{j=1}^{n_2}
            \frac{1}{h_t}K_1\left(\frac{{t_i} - g({s_j})}{h_t}\right)
            \frac{1}{h_y}K_2\left(\frac{{y_{1i}} - y_{2j}}{h_y}\right)
            }
            {
            \displaystyle\frac{1}{n_1 n_2}\displaystyle\sum\limits_{i=1}^{n_1}
            \displaystyle\sum\limits_{j=1}^{n_2}
            \frac{1}{h_t}K_1\left(\frac{{t_i} - g(s_j)}{h_t}\right)},
        \end{equation}
        where $n=n_1+n_2$, $K_1$ and $K_2$ are kernel functions that are probability densities, and $h_t$ and  $h_y$ are the corresponding bandwidths. The above functional can be interpreted as a weighted sum of the terms $(1/h_y)K_2\{(y_{1i} - y_{2j})/{h_y}\}\ (i=1,\ldots, n_1;\ j=1,\ldots, n_2)$, which can be seen as matching scores between $y_{1i}$'s and $y_{2j}$'s. The weights depend on~$g$. Note that when $g=g_0$, for every pair of $i$ and $j$ such that $g(s_j)$ is close to $t_i$, the continuity of $m$ ensures that $y_{1i} - y_{2j}$ is expected to be small. Therefore, for $g=g_0$, large values of $(1/h_y)K_2\{(y_{1i} - y_{2j})/{h_y}\}$ are expected to occur together with large values of their weights. This may not be the case when $g\neq g_0$. In \S~\ref{sec:con}, we show that under some general conditions, the probability limit of $L_n(g)$ attains its maximum value if and only if $g=g_0$.

        Thus, $L_n(g)$ may be interpreted as a measure of alignment or a kernel-matched overall score. This interpretation is clearer from the following equivalent representation of the measure, which can be easily verified:
        \begin{equation}
            L_n(g)=\lim_{\delta\downarrow0}\frac1{\delta}pr\left((y_1(t)-y_2(s))\in [h_yZ_2,h_yZ_2+\delta]\Bigm|t-g(s)=h_tZ_1\right),
            \label{Ln_interpret}
        \end{equation}
        where $(t,y_1(t))$ and $(s,y_2(s))$ are random samples from the empirical distributions of \linebreak $(t_1,y_{11}),\ldots(t_{n_1},y_{1n_1})$ and $(s_1,y_{21}),\ldots(s_{n_2},y_{2n_2})$, respectively, and $Z_1$ and $Z_2$ are random samples from the densities $K_1$ and $K_2$, respectively.

        Let us now examine the roles of the bandwidth parameters $h_t$ and $h_y$ in the above measure. A small value of $h_t$ makes the weight for a given $i$ and $j$ nearly equal to zero, unless $t_i$ is very close to $g(s_j)$. Thus, only a few weights can be substantial. When $h_t$ is large, weights can be substantial for more combinations of $i$ and $j$. Thus, $h_t$ controls the effective number of weights in the weighted sum in \eqref{eq:Ln}. On the other hand, $h_y$ controls the penalty for discrepancies between $y_{1i}$ and $y_{2j}$. A very large value of $h_y$ might make $L_n(g)$ insensitive to changes in $g$, as there would not be enough penalty for mismatch between $y_{1i}$ and $y_{2j}$. A very small value of $h_y$ would make $L_n(g)$ unstable, as $(1/h_y)K_2\{(y_{1i} - y_{2j})/h_y\}$ would be nearly zero for most of the combinations of $i$ and $j$.

        We define the proposed estimator of the function $g_0$ as
        \begin{equation}\label{eq:gn_hat}
            \hat{g}_n=\arg\max_{g\in\G_0} L_n(g),
        \end{equation}
        where $L_n(g)$ is as defined in \eqref{eq:Ln} and $\G_0$ is a suitable class of continuous and strictly increasing functions that includes the true transformation function~$g_0$.

        As for choices of the bandwidths, one can select $h_t$ as a fraction of the average horizontal separation between successive observations in the lesser dense data set and $h_y$ as a fraction of the combined range of the observed variable. Some guidelines are given in \S~\ref{sec:sim_per}.

        The objective function $L_n(g)$ automatically rewards candidate transformation functions that map peaks of one data set into the corresponding peaks of the other. On the other hand, if a peak is missing from one of the data sets, then it does not penalize the `correct' transformation any more than a similar alternative candidate (i.e., a marginally different transformation). Therefore, the estimator $\hat{g}_n$ should be able to utilize these `landmarks' automatically for registration.

    \section{Consistency}\label{sec:con}

        Let the errors $\epsilon_{1i}$ and $\epsilon_{2j}$ and the time points $t_i$ and $s_j\ (i=1,\ldots, n_1;\ j=1,\ldots, n_2)$ in the model~\eqref{eq:model} be mutually independent sets of samples from the probability density functions $f_{\epsilon_1} $, $f_{\epsilon_2} $, $f_1$ and $f_2$ having supports over $(-\infty,\infty)$, $(-\infty,\infty)$, $[a,b]$ (for $a<b$), and $[c,d]$ (for $c<d$) respectively. We can extend the domains of $f_1$ and $f_2$ to the entire real line and define them to be zero outside their respective supports.

        We define $\G$ as the class of all functions $g$ defined over $[c, d]$, differentiable almost everywhere with derivative bounded from below by a positive number, such that the set $S_g=g^{-1}([a, b]) \cap [c, d]$ contains a non-empty open interval and that $g$ agrees with $g_0$ at least at one point in $S_{g_0} \cap S_g$. We first establish the identifiability of $g_0$ (within $\G$) with respect to the model~\eqref{eq:model}.

        \begin{theorem}\label{thm:identifiability}
            Let the class of functions $\G$ be as described above, and assume the following about the functions $m$ and $g_0$ in the model~\eqref{eq:model}.
            \begin{description}
                \item [A1] The function $m$ is continuous. Further, for any given interval $[p,q]$ and for every $y\in\left[\min_{t\in[p,q]}m(t),\max_{t\in[p,q]}m(t)\right]$, the set $\left\{t:\ m(t)=y,\ t\in[p,q]\right\}$ has a finite number of elements.
                \item [A2] The function $g_0$ is a member of the class $\G$.
            \end{description}
            If $g\in\G$ be such that $m(g(s))=m(g_0(s))$ for all $s\in S_g\cap S_{g_0}$, then $S_g=S_{g_0}$ and $g(s)=g_0(s)$ for all $s\in S_{g_0}$.
        \end{theorem}

        The condition~A1 ensures that $m$ is not flat or it does not fluctuate too much.

        The consistency of $\hat{g}_n$ needs to be established as the sample size in both the data sets go to infinity. As a first step, we establish the point-wise convergence of the functionals $L_n$ on $\G$. In order to identify an appropriate probability limit of $L_n(g)$, note that the empirical distributions in \eqref{Ln_interpret} should converge to the corresponding true distributions and $h_t$ and $h_y$ should go to zero, as the sample sizes go to infinity. Therefore, $L_n(g)$ should converge to
        \begin{equation}
            L(g)=\lim_{\delta\downarrow0}\frac1{\delta}pr\left((y_1(t)-y_2(s))\in[0,\delta]\Bigm|t=g(s)\right),
        \label{Lg_interpret}
        \end{equation}
        where $y_1(t)=m(t)+\epsilon_1(t)$ and $y_2(s)=m(g_0(s))+\epsilon_1(s)$ as per the model~\eqref{eq:the_model1} and $t$, $s$, $\epsilon_1(t)$ and $\epsilon_2(s)$ are random samples from the densities $f_1$, $f_2$, $f_{\epsilon_1}$ and $f_{\epsilon_2}$, respectively. This expression simplifies to
        \begin{eqnarray}\label{eq:L}
                L(g)=\frac{\intinfinf\intinfinf f_1(g(y))f_2(y)f_{\epsilon_1}(v-m(g(y))+m(g_0(y)))f_{\epsilon_2}(v) dydv}{\intinfinf f_1(g(y))f_2(y)dy}.
        \end{eqnarray}

        We now formally prove the point-wise convergence of the functionals $L_n$ on $\G$ in the next theorem.

        \begin{theorem}\label{thm:pw_conv}
            Let the class $\G$ be as described at the beginning of this section. Further, let the following assumptions hold in respect of the model~\eqref{eq:model} and the Functional~\eqref{Ln_interpret}.              \begin{description}
            \item [A1*] The function $m$ is continuous.
            \item [A3] The densities $f_{\epsilon_1}$, $f_{\epsilon_2}$, $f_1$, and $f_2$ are continuous and bounded; $f_1$ and $f_2$ are positive over the interior of their supports.\label{asm:1:density:short}
            \item [A4] The kernels $K_1$ and $K_2$ are continuous and bounded probability density functions defined over the real line, and $K_2$ has bounded derivative.\label{asm:2:kernels}
            \item [A5] The sample sizes $n_1$ and $n_2$ and the bandwidths $h_t$ and $h_y$ are such that $n_1/n\rightarrow\xi$ for some $\xi\in(0,1)$, $h_t\rightarrow0$, $h_y\rightarrow0$ and $n^2h_th_y^2\rightarrow\infty$ as $n\rightarrow\infty$.\label{asm:4:h1h2}
            \end{description}
            Then, for any function $g\in\G$,  $L_n(g)$ tends to $L(g)$ in probability as $n\rightarrow\infty$, where $L(g)$ is defined in \eqref{eq:L}.
        \end{theorem}

        The condition~A4 is satisfied by all the popular kernels viz. uniform, triangular, Epanechnikov, biweight, Gaussian, and so on.

        We now show that the limiting functional $L(g)$ is maximized only by the correct transformation function.

        \begin{theorem}\label{thm:uniqueness}
            Suppose the class of functions $\G$ is as described at the beginning of this section, and the Assumptions~A1 and~A2 hold, along with
            \begin{description}
                \item[A3*] The densities $f_{\epsilon_1}$, $f_{\epsilon_2}$, $f_1$, and $f_2$ are continuous and bounded; $f_{\epsilon_1}$ and $f_{\epsilon_2}$ are symmetric about zero and are strictly unimodal at zero; $f_1$ and $f_2$ are positive over the interior of their supports.
            \end{description}
            Then,
            \begin{enumerate}
                \item[(a)] $L(g)\le L(g_0)$ for all $g\in\G$,
                \item[(b)] If $L(g)= L(g_0)$ for some $g\in\G$,
                then $g=g_0$ over $S_{g_0}$.
            \end{enumerate}
        \end{theorem}

        The next step is to establish the uniform convergence of $L_n$, for which we need a stronger condition on $\G$ that enforces compactness.
        Let us define a metric on $\G$ viz., $\Delta(g_1, g_2)=\sup_{x\in[c,d]}|g_1(x)-g_2(x)|=\|g_1-g_2\|$.

        \begin{theorem}\label{thm:unif_conv}
            Let the class of functions $\G$ be as described at the beginning of this section. Suppose the class $\G_0$ in~\eqref{eq:gn_hat} is a compact subset of $\G$ in the metric space $(\G, \Delta)$.\label{asm:5:G} Then under Assumptions~A1*, A3, A5 and the additional assumption
            \begin{description}
                \item[A4*] The kernels $K_1$ and $K_2$ are bounded probability density functions defined over the real line and bounded away from zero on a given closed interval and have bounded first order derivatives,\label{asm:2A:kernels2}
            \end{description}
            the quantity
            $$\sup_{g\in\G_0}\left|L_n(g)-L(g)\right|$$
            tends to zero in probability as $n\rightarrow\infty$.
        \end{theorem}

        An example of $\G_0$ that satisfies the requirements stated in the theorem is the subset of functions $g$ of $\G$ with bounded slope. Assumption~A4* is satisfied by, e.g., the Gaussian kernel.

        We now establish that the sequence of maximum values of the functionals $L_n$ converges to the value of $L$ at its maximizer, $g_0$.

        \begin{theorem}\label{thm:Lngn_hat_conv}
            Let the class of functions $\G$ be as described at the beginning of this section. Suppose the class $\G_0$ in~\eqref{eq:gn_hat} is a compact subset of $\G$ in the metric space $(\G, \Delta)$. Then under Assumptions~A1*, A3*, A4*, A5 and the additional assumption
            \begin{description}
                \item[A2*] The function $g_0$ is a member of the class $\G_0$,
            \end{description}
             the quantity $L_n(\hat{g}_n)$ tends to $L(g_0)$ in probability as $n\rightarrow\infty$.
        \end{theorem}

        We are now ready to establish the consistency of our estimator through the final theorem.

        \begin{theorem}\label{thm:consistency}
            Let the class of functions $\G$ be as described at the beginning of this section. Suppose the class $\G_0$ in~\eqref{eq:gn_hat} is a compact subset of $\G$ in the metric space $(\G, \Delta)$. Then, under the Assumptions~A1, A2*, A3*, A4* and A5, $\hat{g}_n$ tends to $g_0$ in probability as $n\rightarrow\infty$.
        \end{theorem}

    \section{Standard error of the estimator\label{sec:se_g_est}}
        The estimator $\hat{g}_n$ is obtained by maximizing the functional $L_n$, while the correct warping function $g_0$ is the maximizer of the limiting functional $L$. By making use of these facts, one can write $\hat{g}_n(t)-g_0(t)$ in terms of the two kernel functions and the deviations between the empirical distributions of different parts of the data with their theoretical counterparts. \citet{Wang_Gasser_1997} used a similar representation of their estimator of $g_0$ to establish its asymptotic normality, where the scaling factor is of the order of $2/7$ power of the sample size. Noting the slow rate of convergence that is expected, we do not follow this path. Instead we recommend a model-based bootstrap mechanism for estimating the distribution of $\hat{g}_n$ and its various attributes, including the standard error.

        The basis of such a bootstrap scheme would be the model~\eqref{eq:model}, with the functions $g_0$ and $m$, as well as the underlying distributions replaced by suitable estimates. The warping function $g_0$ may be estimated by $\hat{g}_n$. For estimation of $m$, one can use the Nadaraya-Watson estimator from the pooled data set, pretending as if $\hat{g}_n$ is the correct warping function. (Consistency of this estimator has been established in~\citet{Bhaumik_Sengupta_2017}) One has to be careful about generation of the time samples, as sampling with replacement from the finite set of observed time values might produce a large number of ties. A possible strategy is to use kernel density estimators of $f_1$ and $f_2$. The same strategy may be adopted for $f_{\epsilon_1}$ and $f_{\epsilon_2}$, by using the residuals in the two samples as proxies of the respective model errors.

    \section{Simulation of performance}\label{sec:sim_per}

        \subsection{Methods compared}\label{subsec:met_com}

            Even though we have proposed a general class of estimators in Section~\ref{sec:mod_met} and established their consistency in Section~\ref{sec:con}, we need to focus on a specific class in order to simulate the performance. We chose $\G_0$ as the vector space generated by linear B-spline basis functions with equidistant knot points over $[a, b]$. In order to ensure that all the weights in~\eqref{eq:Ln}, which are the ratios of kernels involving $g$, are defined for every trial $g$ in the iterative maximization process, we chose the non-vanishing Gaussian kernel for $K_1$. The other kernel $K_2$ was also chosen as Gaussian. We chose $h_t$ as half of the average horizontal separation between successive observations in the lesser dense data set and $h_y$ as 10 per cent of the combined range of $y$-values of the two data sets. The criterion~\eqref{eq:Ln} was maximized through one-dimensional grid search, rotated over the parameters till convergence. The iterations were continued till the relative change in the criterion over consecutive steps was less than 0.01 per cent. The identity map was used as initial iterate for the maximization algorithm.

            We carried out simulations to compare the performance of the above implementation of the proposed kernel-matched registration method with two other methods.
            \begin{enumerate}
                \item	The first method for comparison was continuous monotone registration. The \texttt{fda} package of \texttt{R} was used for computations, which involved two steps. In the first step, smooth curves were computed from noisy data sets by using the routine \texttt{smooth.basis}, while the warping function was estimated in the second step by using \texttt{register.fd}. Both the computations used a common B-spline basis system which was specified by the routine \texttt{create.bspline.basis} with default choice of order of b-spline basis functions~(4) and number of equidistant break-points same as the number of knot-points used for the proposed method. Functions were estimated by penalizing curvatures with a common choice of the smoothing parameter ($\lambda$) as $10^{-4}$, as suggested in \cite{Ramsay_Li_1998}.
                \item	The second method was self-modelling registration \citep{Gervini_Gasser_2004}. MATLAB codes provided by the authors were used for computations with default choices for the number of random starts (20) and the order of splines (3). The sets of data chosen for simulations revealed the existence of approximately seven prominent features. Therefore, following suggestions by the authors, the plausible value of number of components, $q$, was chosen to vary from 7 to 12, which also satisfied the prescribed rule of thumb (i.e.\ $q <$ square root of number of observations). The suggested possible values of the number of basis functions, $p$, were $3q$ or $4q$. The final choice was made by the cross-validation algorithm proposed by the authors.
            \end{enumerate}
            These methods were selected for comparison mostly on account of applicability to the data at hand and availability of codes. The wavelet based automated method of \cite{Bigot_2006} requires size of the input data sets to be equal to powers of 2, which made it unsuitable for the data we considered.

        \subsection{Simulation design}\label{subsec:sim_des}

            The function $m$ chosen for simulation, shown in Figure~\ref{fig:sim_m_fun}, has sharp movements similar to those of the paleoclimatic series of carbon dioxide. The distributions $f_{\epsilon_1}$ and $f_{\epsilon_2}$ were chosen as normal with mean 0 and standard deviation equal to $0.05s$, where $s^2=(1/T)\int_0^T(m(t)-\overline{m})^2dt$, $\overline{m}=(1/T)\int_0^T m(t)dt$ and $T=400$ (upper end of the time scale in Figure~\ref{fig:sim_m_fun}). We conducted Monte Carlo simulations (with $n_1=n_2=250$) in four different scenarios of $g_0$ and time sample selection as described below:

            \begin{description}
                \item[Scenario 1.] For all simulation runs, the sets of time samples of the two data sets were kept fixed and identical to one another. This set consisted of an uniformly spaced sample of size 250 chosen from the interval $[0, 400]$, inclusive of the end-points. The time transformation function $g_0$ in~\eqref{eq:model}, was chosen as the linear spline,
                    \begin{eqnarray}
                        g_0(t)&=&0.95t + 0.2(t-80)_{+} - 0.4(t-160)_{+} \nonumber\\
                        && + 0.6(t{-}240)_{+} - 0.55(t{-}320)_{+} \quad t\in [0, 400]\label{sim_g_1}
                    \end{eqnarray}
                    where $u_{+}$ is $u$ for positive $u$ and zero otherwise;

                \item[Scenario 2.] The $250$ time points were generated afresh for each simulation run as samples from uniform distribution over $[0, 400]$, separately for the two data sets. The function $g_0$ was chosen as in~\eqref{sim_g_1}.

                \item[Scenario 3.] The time samples were chosen as in Scenario~1, while $g_0$ was chosen as,
                    \begin{eqnarray}
                        g_0(t)=t+0.05t\sin\left(\frac{4\pi t}{400}\right)\quad t\in [0, 400].\label{sim_g_2}
                    \end{eqnarray}

                \item[Scenario 4.] The time samples were chosen as in Scenario~2, while $g_0$ was chosen as in~\eqref{sim_g_2}.
            \end{description}

            \begin{figure}
                \begin{center}
                    \includegraphics{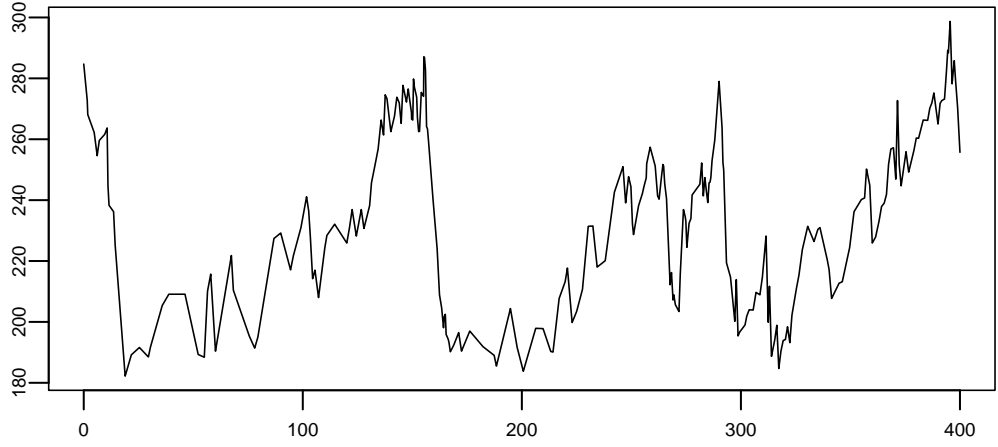}
                    \caption{The function $m$ for simulation}\label{fig:sim_m_fun}
                \end{center}
            \end{figure}

            \begin{figure}
                \begin{center}
                    \includegraphics{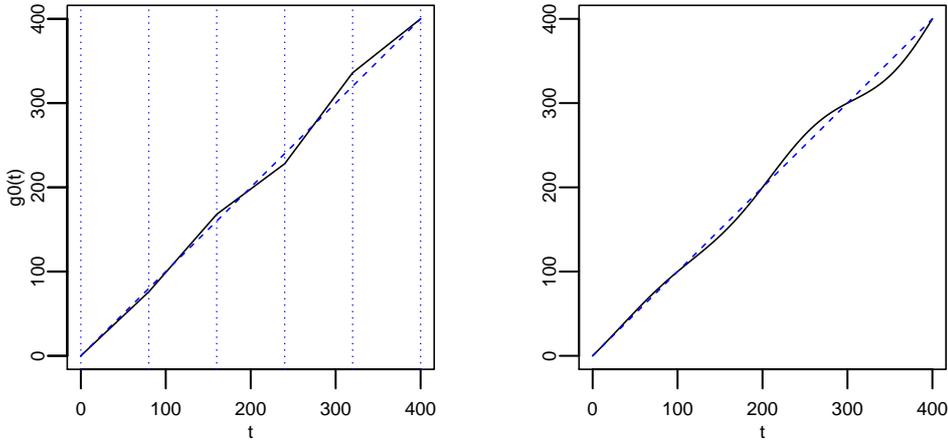}
                    \caption{Time transformation functions $g_0$ (solid lines) used in simulations and identity map (dashes): linear spline (left) and periodic (right); vertical dots indicate locations of the knots of linear spline}\label{fig:sim_g0}
                \end{center}
            \end{figure}

            For all the scenarios of simulation and the data analysis reported in the next section, the number of line segments in the search space of B-splines for the proposed algorithm was chosen as 10 times the fifth root of the size of the smaller data set (rounded to the nearest integer), as prescribed in \citet[pages 303--304]{Eubank_1999}. The same number of polynomials was used for continuous monotone registration as well. For the proposed algorithm in Scenario 1, this choice implies that $\G_0$ consists of linear B-splines with 30 line segments, which includes~\eqref{sim_g_1}.

            The MATLAB programs for self-modelling registration require values of the functions to be registered at a common set of time points, which is met only in scenario $1$ and $3$. Therefore, simulation results for this method are reported only for these two scenarios. In fact, these two scenarios were used in the simulations only to enable implementation of this method.

            The performance of the estimators of $g_0$ were studied in terms of (a)~point-wise bias, (b)~point-wise standard deviation, (c)~point-wise mean squared error, and (d)~average of the integrated mean squared error normalized by the squared norm of the true function, defined for each simulation run as
            \begin{equation}\label{imse}
                \frac{\frac{1}{S}\sum_{j=1}^{S}\int_0^T(\hat{g}_j(t)-g_0(t))^2dt}{\int_0^T g_{0}^{2}(t)dt}
            \end{equation}
            where $S$ was the number of independent runs of the simulation and $\hat{g}_j$ is the estimate of $g_0$ at the $j$th run. We used Simpson's rule to evaluate these definite integrals.

        \subsection{Results}\label{subsec:sim_res}
            The point-wise bias, standard deviation, and mean squared error of the three estimators of $g_0$, estimated from $1000$ simulation runs for each of the scenarios, are shown in Figure~\ref{fig:bias_sd_mse}.  The available implementation of Self-modelling registration works only for matched time points in the two data sets, as mentioned earlier. For this reason, it could be used only for Scenarios~1 and~3. Further, it did not converge for 13 simulation runs in Scenario 1 and 20 runs in Scenario~3. These runs were excluded from the empirical evaluations of performance.
            It is observed that the proposed estimator generally has smaller bias, standard deviation and mean squared error as compared to the other methods.
            The average normalised integrated mean squared errors of the three methods are summarized in Table~\ref{tab:imse}. The kernel-matched registration method is found to have uniformly better performance across all the scenarios.

            \begin{table}
                \caption{\label{tab:imse}Average normalized integrated mean squared error of three estimators, normalized by the squared norm ($\times\ 10^{-3}$)}
                \centering
                \fbox{%
                    \begin{tabular}{lrrrr}
                        \hline
                         Method&Scenario\ 1&Scen.\ 2&Scen.\ 3&Scen.\ 4\\
                         \hline
                         Cont.\ mon.\ registration&0.223&0.459&0.851&1.138\\
                         Self-modelling registration&0.295&--&1.303&--\\
                         Kernel-matched registration&0.004&0.208&0.003&0.338\\
                         \hline
                    \end{tabular}}
            \end{table}

            \begin{figure}
                \begin{center}
                    \includegraphics{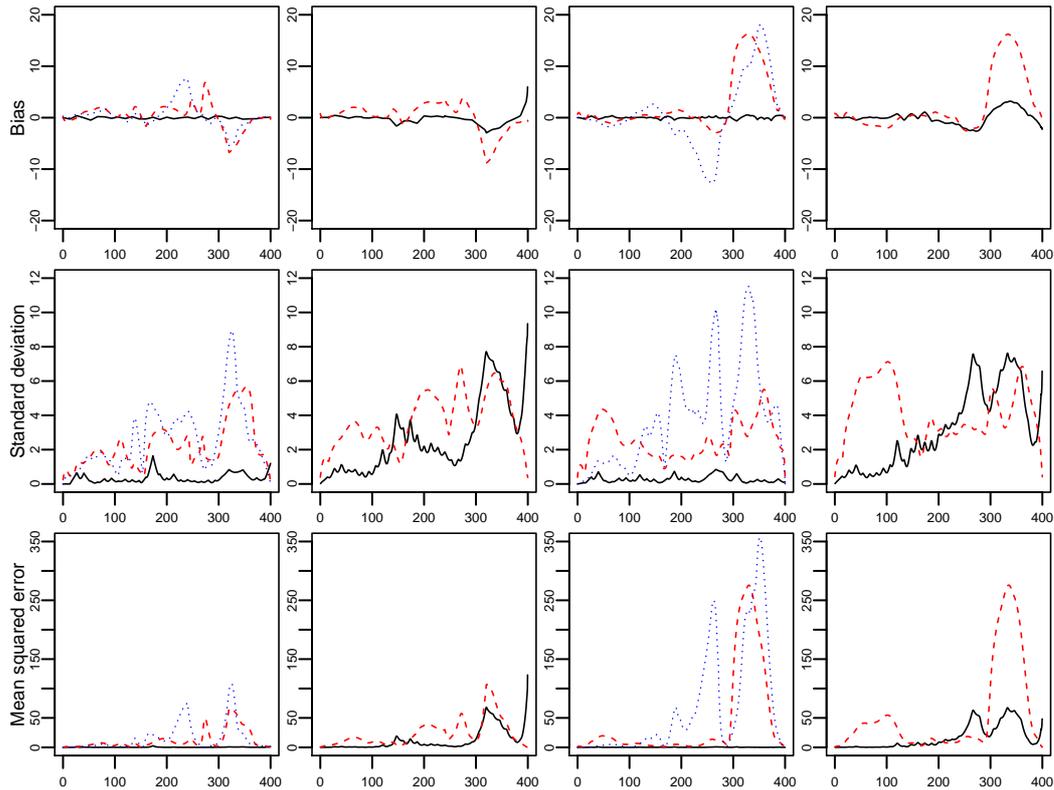}
                    \caption{Simulated point-wise bias, standard deviation, and mean squared error of the estimators of $g_0$ by continuous monotone registration (dashes), self-modelling registration (dots) and kernel-matched registration (solid lines) under Scenarios~1 (first column), 2~(second), 3~(third) and 4~(last column)}\label{fig:bias_sd_mse}
                \end{center}
            \end{figure}

            The relatively large standard deviation of the proposed estimator in Scenarios~2 and 4 may have resulted from the grid search having missed the global maximum of the criterion $L_n(g)$. In order to investigate this possibility, we computed the functional $L_n(g)$ at the correct value $g_0$ of the warping function. Figure~\ref{fig:cri_curves} shows the scatter plot of $L_n(g_0)$ vs.\ $L_n(\hat{g}_n)$ across the 1000 simulation runs under the four scenarios. The points above the diagonal correspond to the cases where $L_n(g_0)$ is larger than $L_n(\hat{g}_n)$, i.e.,where the global maximum is surely missed. Such occurrence are found to occur more in Scenarios~2 and 4, where the time points are chosen randomly. It transpires that a better search algorithm (e.g., simulated annealing) might reduce the standard deviation of the proposed estimator.

            \begin{figure}
                \begin{center}
                    \includegraphics{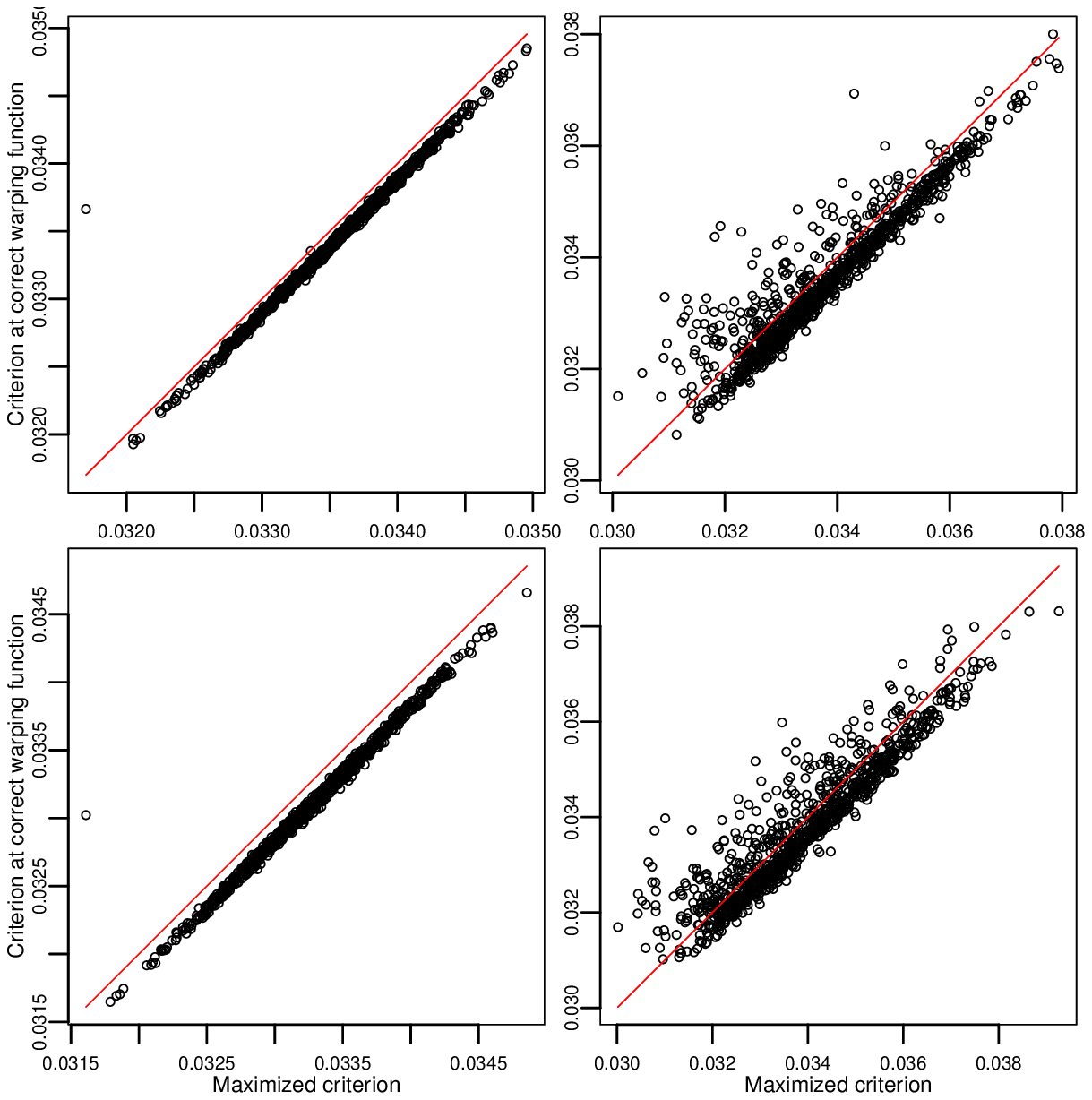}
                    \caption{Scatter-plot of criterion value at the correct warping function verses ``maximized'' value of the criterion from 1000 simulation runs under Scenarios~1 (top-left), 2~(top right), 3~(bottom-left) and 4~(bottom right)}\label{fig:cri_curves}
                \end{center}
            \end{figure}

    \subsection{Appropriateness of standard error}

    We ran an additional set of simulations to check the computation of standard error through model based bootstrap, described in Section~\ref{sec:se_g_est}. After generating data under Scenario~4, we computed, for a given simulation run, a bootstrap estimate of the standard deviation of $\hat{g}_n$ with the following choices in \eqref{eq:model}: (i)~the warping function~$g_0$ is replaced by $\hat{g}_n$, (ii)~the function $m$ is replaced by the estimator
    \begin{equation*}
        m_{n,\hat{g}_n}(t)=m_{n,g}(t)\big|_{g=\hat{g}_n},
    \end{equation*}
    where
    \begin{equation*}
        \displaystyle m_{n,g}(t)=
        \frac
        {
        \displaystyle\frac{1}{nh_n}\left\{\displaystyle\sum_{i=1}^{n_1}K\left(\frac{t-t_i}{h_n}\right) y_{1i}+\displaystyle\sum_{j=1}^{n_2}K\left(\frac{t-g(s_j)}{h_n}\right)y_{2j}\right\}
        }
        {
        \displaystyle\frac{1}{nh_n}\left\{\displaystyle \sum_{i=1}^{n_1}K\left(\frac{t-t_i}{h_n}\right) +\displaystyle\sum_{j=1}^{n_2}K\left(\frac{t-g(s_j)}{h_n}\right)\right\}
        }
    \end{equation*}
    with Gaussian kernel chosen for $K$ and bandwidth chosen by leave-one-out cross-validation,  (iii)~the densities $f_1$ and $f_2$ are replaced by their kernel density estimates based on the observed data and (iv)~the densities $f_{\epsilon_1}$ and $f_{\epsilon_2}$ are replaced by their kernel density estimates based on the residuals in the two samples. For the four density estimates, we used the Gaussian kernel and chose the bandwidth by likelihood cross-validation. One thousand bootstrap samples were generated from a given simulation run. Each bootstrap pair of samples was of size 250 (same as the original data). The sample standard deviation of the kernel matched registration estimator computed from the 1000 re-samples from a single simulation run should be comparable with the sample standard deviation of that estimator computed from 1000 runs, plotted in the middle right plot of Figure~\ref{fig:bias_sd_mse}. Figure~\ref{fig:btstrp_se_g_est} shows the plots of these bootstrap estimators from five different simulation runs (all of them as dashed curves), with the sample standard deviation obtained from 1000 simulation runs (solid curve) used as benchmark. It is seen that the bootstrap estimates are reasonable.

    \begin{figure}
        \begin{center}
            \includegraphics{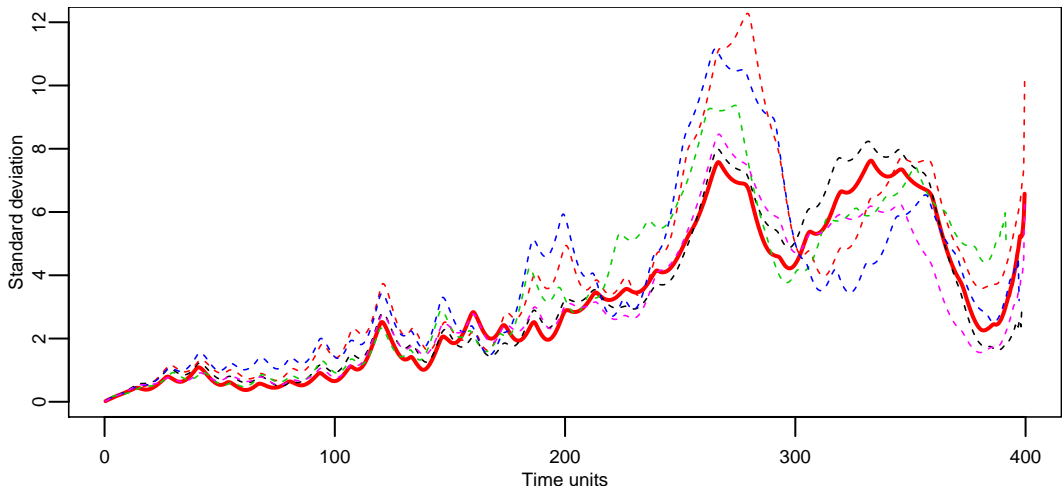}
            \caption{Bootstrap estimates of standard errors of the estimated warping function produced by kernel-matched registration from five different simulation runs (dashes), with the empirical standard deviation obtained from 1000 simulation runs (thick solid line)}\label{fig:btstrp_se_g_est}
        \end{center}
    \end{figure}

   \section{Analysis of ice core data}\label{sec:dat_ana}

        We considered paleoclimatic data on the atmospheric concentration of (i) carbon dioxide and (ii) methane \citep{Petit_et_al_1999, Loulergue_et_al_2008} as determined from air-bubbles trapped in ice cores collected over Lake Vostok and at EPICA Dome in Antarctica and (iii)~average annual temperature deviations \citep{Petit_et_al_1999, Valerie_2007}, which were reconstructed from deuterium contents at various depths of ice cores obtained at these two sites. Table~\ref{tab:des_stat} gives an idea of the volume and magnitude of the variables of these data sets.
        \begin{table}
            \caption{\label{tab:des_stat}Some descriptive statistics of the data sets}
            \centering
            \fbox{%
                \begin{tabular}{lrcrc}
                     & \multicolumn{2}{c}{Data set 1: Vostok} & \multicolumn{2}{c}{Data set 2: EPICA Dome} \\
                     Data&Size&Range($Y$-Value)&Size&Range($Y$-Value)\\
                     \hline
                     Carbon dioxide&283&182.2-298.7&537&183.8-298.6\\
                     Methane&457&318-773&1,545&342-907\\
                     Temp. deviations&3,310&(-)9.39-3.23&5,028&(-)10.58-5.46\\
                \end{tabular}}
        \end{table}

        We chose to align the data set from Lake Vostok with that from EPICA Dome C by using two registration methods namely, the proposed kernel-matched registration and continuous monotone registration. Parameters of these methods for data analyses were chosen as in the previous section. We could not apply the method of self-modelling registration for the above data sets, since it requires the nominal observation times in the two data sets to coincide.

        Alignments of the three pairs of data sets, produced by the two methods, are plotted in Figures~\ref{fig:co2_aligned}--\ref{fig:temp_aligned} and  the estimates of $g_0$ are plotted in Figure~\ref{fig:est_g}. Kernel-matched registration is found to produce visibly better alignment.
        \begin{figure}
            \begin{center}
                \includegraphics{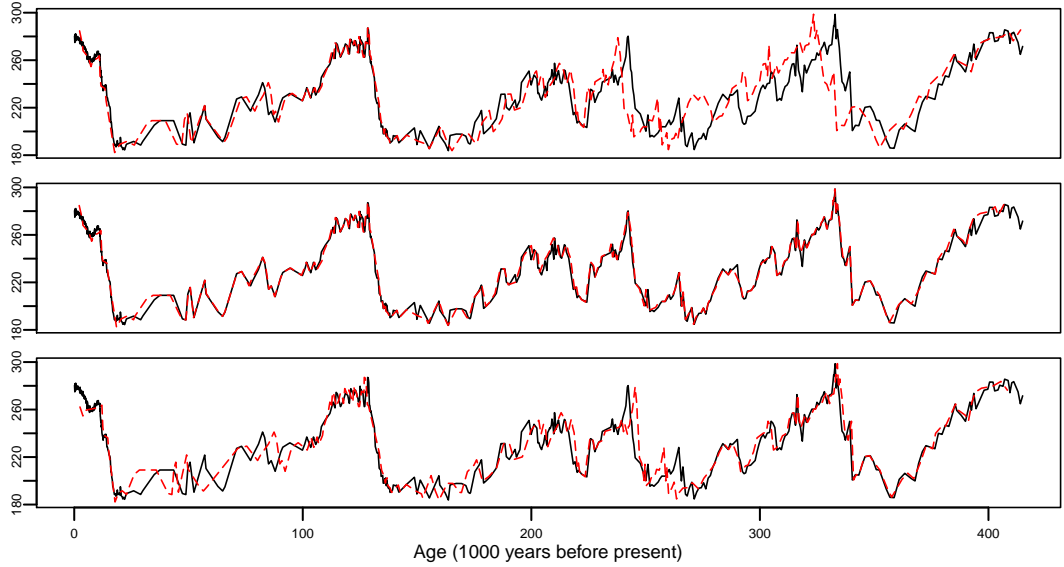}
                \caption{Atmospheric concentration of carbon dioxide (in parts per million volume) constructed from EPICA Dome C (solid lines) and Lake Vostok (dashes) ice cores (top); alignment produced by the kernel-matched registration (middle) and continuous monotone registration (bottom)}\label{fig:co2_aligned}
            \end{center}
        \end{figure}

        \begin{figure}
            \begin{center}
                \includegraphics{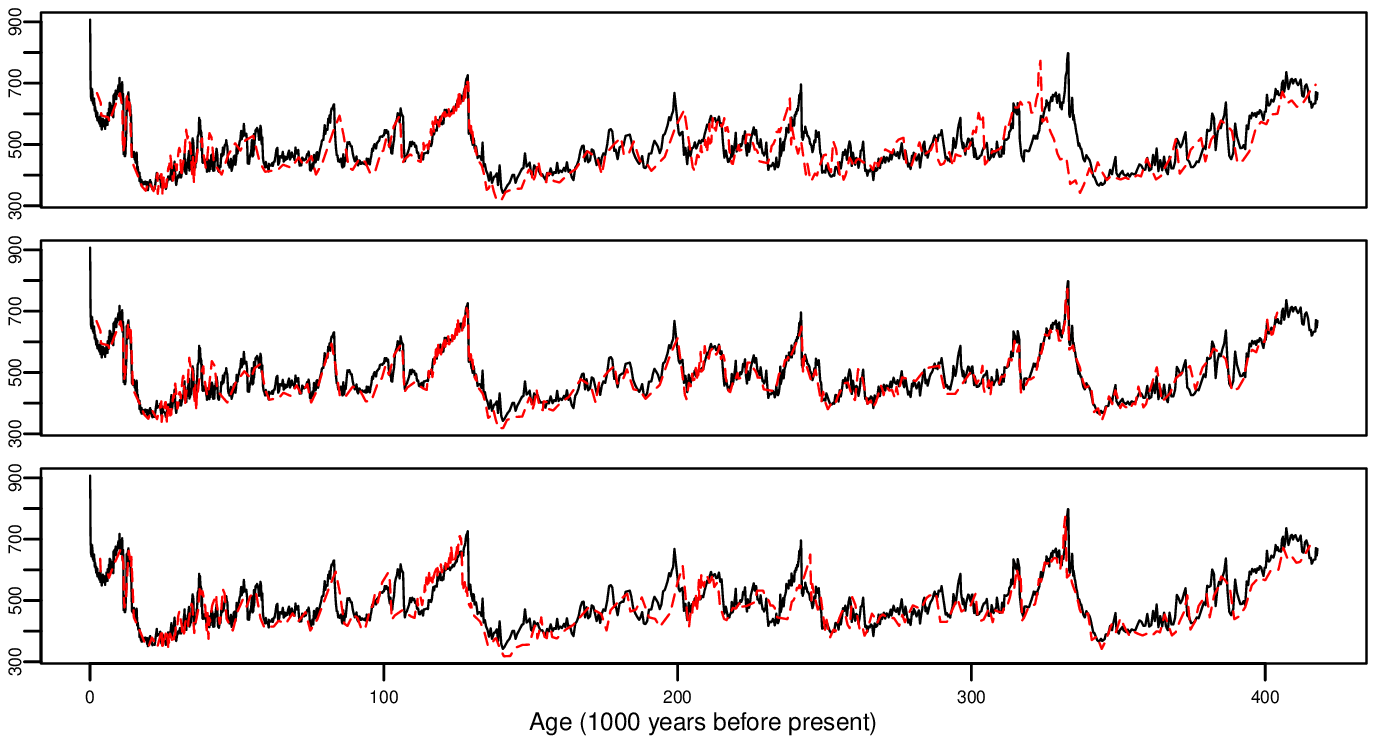}
                \caption{Atmospheric concentration of methane (in parts per billion by volume) constructed from EPICA Dome C (solid lines) and Lake Vostok (dashes) ice cores (top); alignment produced by the kernel-matched registration (middle) and continuous monotone registration (bottom)}\label{fig:ch4_aligned}
            \end{center}
        \end{figure}

        \begin{figure}
            \begin{center}
                \includegraphics{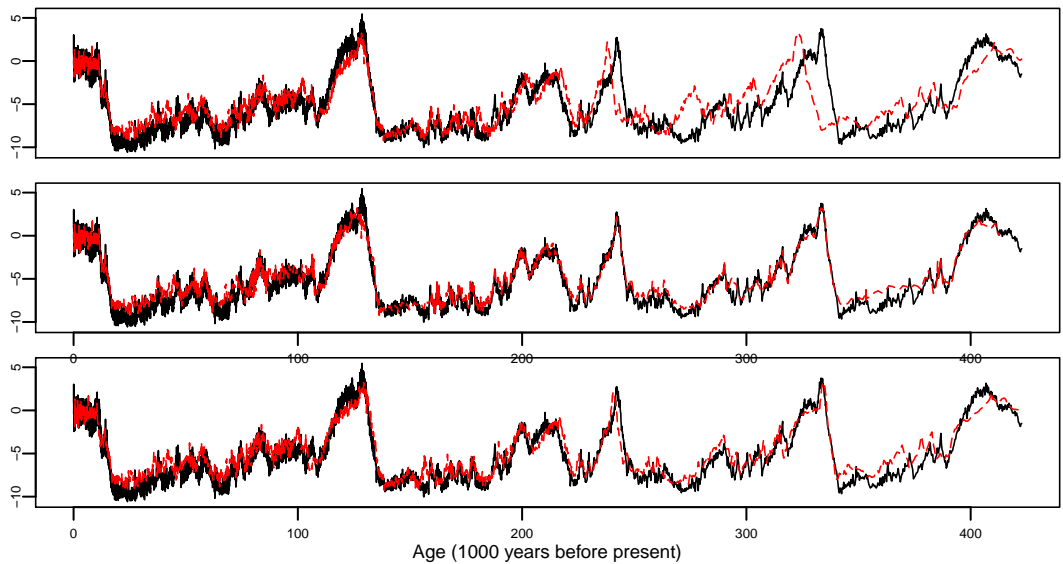}
                \caption{Average annual temperature deviations (in degrees celcius) constructed from EPICA Dome C (solid lines) and Lake Vostok (dashes) ice cores (top); alignment produced by the kernel-matched registration (middle) and continuous monotone registration (bottom)}\label{fig:temp_aligned}
            \end{center}
        \end{figure}

        \begin{figure}
            \begin{center}
                \includegraphics{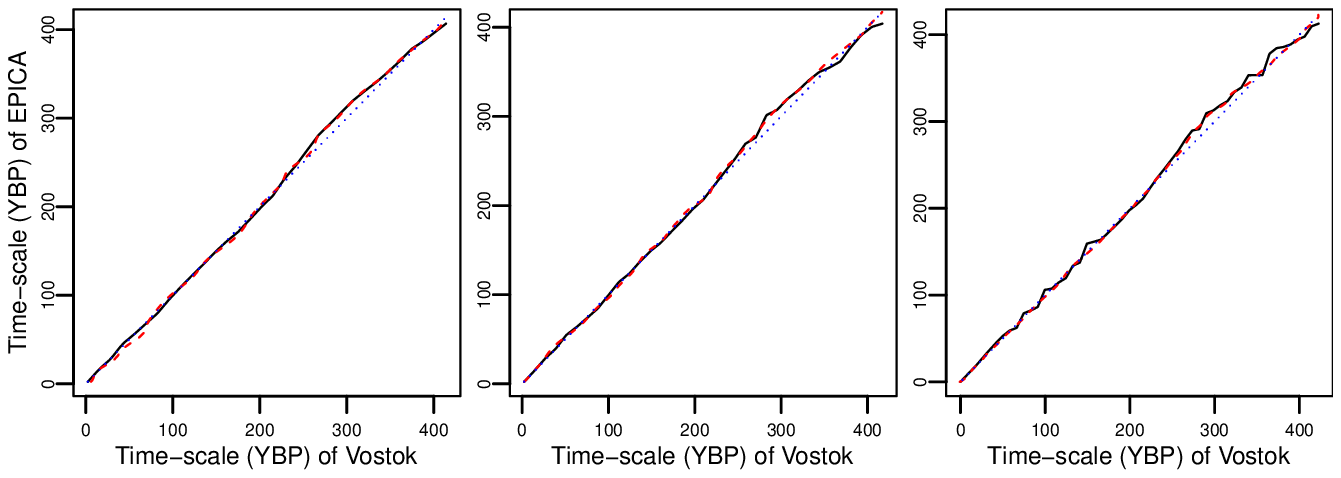}
                \caption{Estimates of $g_0$ for carbon dioxide (left), methane (middle), and temperature (right) data by kernel-matched registration (solid lines) and continuous monotone registration (dashes) against identity (dot)}\label{fig:est_g}
            \end{center}
        \end{figure}

        The average squared distance between the linearly interpolated data sets before and after registration by the two methods, computed over a grid of size 1000 over the time range, are reported in Table~\ref{tab:sq_diff}. It is observed that the kernel-matched registration method produced shorter distance between the registered curves as compared to continuous monotone registration.
        \begin{table}
            \caption{\label{tab:sq_diff}Average squared difference between the pairs of data sets}
            \centering
            \begin{tabular}{lrrr}
                \hline
                &Carbon dioxide&Methane&Temp.\ dev.\ \\
                \hline
                Pre-alignment &229.4&3,859.1&4.7\\
                Post-alignment:&&&\\
                \hspace{0.2in}Continuous monotone registration&99.9&2,147.7&1.8\\
                \hspace{0.2in}Kernel-matched registration&12.7&991.6&1.0\\
                \hline
            \end{tabular}
        \end{table}

    \section{Discussion}\label{sec:dis}

        In this paper, we have proposed a new method of registration of one functional data set with another, by optimizing an empirical kernel-based measure of alignment. If there are sharp features in the data, the proposed method is able to utilize them, without requiring prior identification of landmarks. Since the method does not use any pre-smoothing, it does not suffer from any loss of information that might occur due to smoothing. However, the measure of alignment \eqref{eq:Ln} ensures that the proposed method makes use of the main strength of smoothing, namely pooling of information from neighboring observations. It also makes use of the main strength of marker registration, by rewarding match of sharp features in the two sets of data.

        The present implementation of the method, in the form of R codes, is available from the first author on request. This implementation permits registration of data sets with possibly unequal, irregularly spaced and large number of samples. This implementation is based on some specific choices, e.g., use of the class of B-splines with uniformly spaced knot points as candidate time transformation functions, and steepest ascent for optimization. However, none of these choices is necessary in the general set-up used for proving the consistency of the proposed class of estimators.

        There are indeed some limitations of kernel-matched registration. For very noisy data sets, a sharp peak in one data set may not have a matching peak in the other one. Such spurious peaks might confuse this method. The proposed method would also be unsuitable for longitudinal growth data with many individuals but relatively fewer observations per individual.

        The alignment provided by the proposed method is likely to change if the data sets for registration are interchanged. One can choose the solution that produces smaller average of squared difference between the pair of registered and linearly interpolated data sets. When there are more than two data sets to be registered, the proposed method has to be used multiple times on pairs of data, possibly after identifying one of the data sets as reference for registration. This reference data set may also be selected on a trial basis, and the candidate leading to the best overall alignment may be selected as reference data set at the end. Alternatively, the two-step approach suggested by \citet{Tang_Muller_2008} may be used.

        The proposed method can be used as a tool for estimating the function $m$ in model~\eqref{eq:model}. Large sample properties of the resulting estimator of $m$ and its performance in relation to competing methods are studied in~\citet{Bhaumik_Sengupta_2017}.

        Registration of two sets of functional data on different variables (e.g., paleoclimatic data on temperature and carbon dioxide) is sometimes needed for the purpose of studying the relationship between them. The method presented here can be made applicable to this problem, provided the two variables are approximately related through a linear transformation. This problem is being studied by the authors.

    \appendix

    \section{Appendix: Proofs of theoretical results}

    \mbox{}\indent{\sc Proof of Theorem~\ref{thm:identifiability}.}
        Let $g\in\G$ be such that $g(\alpha)= g_0(\alpha)$ for some $\alpha\in S_{g_0}\cap S_g$, $m(g(s)) = m(g_0(s))$ for all $s\in S_{g_0}\cap S_g$, and yet $g(s)\ne g_0(s)$ for some $s\in S_{g_0}\cap S_g$.  We can presume, without loss of generality, that $g(s)<g_0(s)$. Let us assume, for now, $\alpha<s$. Then the set $\{t:\ m(t)=m(g_0(s))\}\cap[g_0(\alpha), g_0(s)]$ has at least two elements, $g(s)$ and $g_0(s)$. Let $t^{(1)}<\cdots<t^{(k)}$ be the ordered elements of this set. Clearly, $g_0(s)=t^{(k)}$. Let $g(s)=t^{(i)}$ for some $i<k$.

        In order that the functions $m(g(u))-m(g(s))$ and $m(g_0(u))-m(g_0(s))$ coincide for all $u\in[\alpha,s]$, these functions should have exactly the same number of zero crossings over this interval. However, from what we have already observed, the first function has exactly $i$ zero-crossings, while the second function has exactly $k$ zero-crossings, and $i<k$. Therefore, the two functions must differ somewhere on $[\alpha,s]$.

        Similarly, if $\alpha > s$, the set, $\{t:\ m(t)=m(g(s))\}\cap [g(s), g(\alpha)]$ has at least two elements viz., $g(s)$ and $g_0(s)$. If $s^{(1)}<\cdots <s^{(l)}$ be the ordered elements of this set, then $g(s)=s^{(1)}$ and $g_0(s)=s^{(j)}$ for some $j>1$. By similar arguments, the  two functions namely $m(g(u))-m(g(s))$ and $m(g_0(u))-m(g_0(s))$, which have been presumed to coincide for all $u\in[s, \alpha]$, must differ somewhere on $[s, \alpha]$ as the first function has exactly $l$ zero-crossings, while the second has exactly $l-j+1$.

        This contradicts the presumption that $g(s)\ne g_0(s)$ for some $s\in S_g\cap S_{g_0}$. Thus, $g=g_0$ over $S_g\!\cap\!S_{g_0}$.

        The continuity of $g$ and $g_0$, together with their equality over $S_g\cap S_{g_0}$, implies that $S_g=S_{g_0}$.

\bigskip
            {\sc Proof of Theorem~\ref{thm:pw_conv}.}
            Write $L_n(g)$ and $L(g)$ in~\eqref{eq:Ln}~and~\eqref{eq:L} as
            \begin{eqnarray}\label{eq:Ln_L}
            L_n(g)=\frac{N_n(g)}{D_n(g)},\quad L(g)=\frac{N(g)}{D(g)}
            \end{eqnarray}
            where
            \begin{eqnarray}
                N_n(g)&=&\frac{1}{n_1 n_2}\displaystyle\sum\limits_{i=1}^{n_1}\displaystyle\sum\limits_{j=1}^{n_2}\frac{1}{h_t}K_1\left(\frac{{t_i} - g({s_j})}{h_t}\right)\frac{1}{h_y}K_2\left(\frac{{y_{1i}} - y_{2j}}{h_y}\right),\label{eq:Nn}\\
                D_n(g)&=&\frac{1}{n_1 n_2}\displaystyle\sum\limits_{i=1}^{n_1}\displaystyle\sum\limits_{j=1}^{n_2}\frac{1}{h_t}K_1\left(\frac{{t_i} - g(s_j)}{h_t}\right),\label{eq:Dn}\\
                N(g)&=&\intinfinf\intinfinf f_1(g(y))f_2(y)f_{\epsilon_1}(v-m(g(y))+m(\g(y)))f_{\epsilon_2}(v) dydv,\notag\\
                \label{eq:N}\\
                D(g)&=&\intinfinf f_1(g(y))f_2(y)dy.\label{eq:D}
            \end{eqnarray}
            We shall show that $N_n(g)$ tends to $N(g)$ and $D_n(g)$ tends to $D(g)$ in probability as $n\rightarrow\infty$, by showing that $E(N_n(g))$ and $E(D_n(g))$ tend to $N(g)$ and $D(g)$, respectively, and $var(N_n(g))$ and $var(D_n(g))$ tend to zero as $n\rightarrow\infty$. The fact $D(g)>0$ follows from Assumption~A3 and the condition that $g([c, d])\cap [a, b]$ includes a non-empty open interval.

            From~\eqref{eq:Nn}, by~\eqref{eq:model} and Assumption~A3 we have
            \begin{eqnarray*}
                &&\hskip-20pt E(\Nn{g})\\
                &=&\intinfinf\intinfinf\intinfinf\intinfinf\frac{1}{h_t}K_1\left(\frac{x-g(y)}{h_t}\right) \frac{1}{h_y}K_2\left(\frac{m(x)-m(\g(y))+u-v}{h_y}\right)\\
                &\times& f_1(x)f_2(y) f_{\epsilon_1}(u)f_{\epsilon_2}(v)dxdydudv\\
                &=&\int_{-\infty}^{\infty}\int_{-\infty}^{\infty}\int_{-\infty}^{\infty}\int_{-\infty}^{\infty} K_1\left(z_1\right)K_2\left(z_2\right)f_1(g(y)+z_1h_t)\\
                &\times& f_2(y)f_{\epsilon_1}(v-m(g(y)+z_1h_t)+m(\g(y))+z_2h_y)f_{\epsilon_2}(v)dz_1dz_2dydv.
            \end{eqnarray*}
            The above integrand is bounded by an integrable function namely $$M_f^2K_1(z_1)K_2(z_2)f_2(y)f_{\epsilon_2}(v),$$ where the real number $M_f$ is the common upper bound of $f_1$, $f_2$, $f_{\epsilon_1}$ and $f_{\epsilon_2}$ (see Assumption~A3). By Assumptions~A1*, A3 and A5, for a fixed $(z_1,z_2,y,v)$, the integrand tends to $K_1(z_1)K_2(z_2)f_1(g(y))f_2(y)f_{\epsilon_1}(v-m(g(y))+m(g_0(y)))f_{\epsilon_2}(v)$ as $n\rightarrow\infty$. Consequently, by dominating convergence theorem and Assumptions~A1*, A3, A4 and~A5,
            \begin{eqnarray}\label{eq:thm:pw_conv:S1:Nng_conv}
                \lim_{n\rightarrow\infty}E(N_n(g))=N(g),
            \end{eqnarray}
            where $N(g)$ is as in~\eqref{eq:N}.

            \medskip
            From~\eqref{eq:Dn}, by~\eqref{eq:model} and Assumption~A3, we have
            \begin{eqnarray}
                E(\Dn{g})&=&\int_{-\infty}^\infty\int_{-\infty}^{\infty}I_{\left(\frac{-g(y)}{h_t},\infty\right)}(z)K_1(z)f_1(g(y)+zh_t)f_2(y)dzdy.\notag
            \end{eqnarray}
            By a similar argument as above, by Assumptions~A3, A4 and A5, we have
            \begin{eqnarray}\label{eq:thm:pw_conv:S2:Dng_conv}
                \lim_{n\rightarrow\infty}E(\Dn{g})=D(g)
            \end{eqnarray}
            where $D(g)$ is as in~\eqref{eq:D}.

            \medskip
            Write $var(N_n(g))=v_1(n)+v_2(n)+v_3(n)+v_4(n)$, where
            \begin{eqnarray}
                v_1(n)&=&\frac{1}{(n_1n_2h_th_y)^2}\sum_{i=1}^{n_1}\sum_{j=1}^{n_2}\!\!var\!\!\left\{\!\!K_1\!\!\left(\frac{t_i- g(s_j)}{h_t}\right)
                \!\!K_2\!\!\left(\frac{y_{1i}-y_{2j}}{h_y}\right)\right\},\label{eq:thm:pw_conv:S3:v1n}\\
                v_2(n)&=&\frac{1}{(n_1n_2h_th_y)^2}\sum_{i=1}^{n_1}\sum_{j=1}^{n_2}\sum^{n_1}_{
                \substack{i'=1 \\ (\neq i)}}cov\left\{K_1\left(\frac{t_i- g(s_j)}{h_t}\right)K_2\left(\frac{y_{1i}-y_{2j}}{h_y}\right)\right.,\notag\\
                &&\qquad \left.K_1\left(\frac{t_{i'}-g(s_{j})}{h_t}\right)K_2\left(\frac{y_{1i'}-y_{2j}}{h_y}\right)\right\},\label{eq:thm:pw_conv:S4:v2n}\\
                v_3(n)&=&\frac{1}{(n_1n_2h_th_y)^2}\sum_{i=1}^{n_1}\sum_{j=1}^{n_2}\sum^{n_2}_{\substack{
                j'=1\\ (\neq j)}} cov\left\{K_1\left(\frac{t_i-g(s_j)}{h_t}\right)K_2\left(\frac{y_{1i}-y_{2j}}{h_y}\right)\right.,\notag\\
                &&\qquad \left.K_1\left(\frac{t_{i}- g(s_{j'})}{h_t}\right)
                K_2\left(\frac{y_{1i}-y_{2j'}}{h_y}\right)\right\},\label{eq:thm:pw_conv:S5:v3n}\\
                v_4(n)&=&\frac{1}{(n_1n_2h_th_y)^2}\sum_{i=1}^{n_1}\sum_{j=1}^{n_2}\sum^{n_1}_{\substack{
                i'=1 \\ (\neq i)}}\sum^{n_2}_{\substack{j'=1\\ (\neq
                j)}} cov\left\{K_1\left(\frac{t_i- g(s_j)}{h_t}\right)
                K_2\left(\frac{y_{1i}-y_{2j}}{h_y}\right)\right.\notag\\
                &&\qquad \left.K_1\left(\frac{t_{i'}- g(s_{j'})}{h_t}\right)
                K_2\left(\frac{y_{1i'}-y_{2j'}}{h_y}\right)\right\}\notag.
            \end{eqnarray}
            From the model specifications, it follows that $v_4(n)=0$. We shall show that each of the other $v_i(n)$'s tends to zero as $n\rightarrow\infty$.

            From~\eqref{eq:thm:pw_conv:S3:v1n}, by~\eqref{eq:model} and Assumption~A3 we have
            $$v_1(n)=\frac{1}{n_1n_2h_th_y}v_{11}(n)-\frac{1}{n_1n_2}E^2(N_n(g))$$ where
            \begin{eqnarray}
                v_{11}(n)&=&\int_{-\infty}^{\infty}\int_{-\infty}^{\infty}\int_{-\infty}^{\infty}\int_{-\infty}^{\infty} K_1^2(z_1)K_2^2(z_2)f_1(g(y)+z_1h_t)\nonumber\\
                &\times&f_2(y)f_{\epsilon_1}(v-m(g(y)+z_1h_t)+m(\g(y))+z_2h_y)f_{\epsilon_2}(v)dz_1dz_2dydv.\notag
            \end{eqnarray}
            By a similar argument as at the beginning of this proof, by Assumptions~A1*, A3, A4 and~A5 we have
            \begin{eqnarray}
                \lim_{n\rightarrow\infty}n_1n_2h_th_yv_{1}(n)&=&\int_{-\infty}^{\infty}K_1^2(z_1)dz_1\int_{-\infty}^{\infty} K_2^2(z_2)dz_2\notag\\
                &&\hskip-45pt \times\int_{ -\infty}^{\infty}\int_{-\infty}^{\infty}f_1(g(y))f_2(y)f_{\epsilon_1}(v-m(g(y))+m(\g(y)))f_{\epsilon_2}(v) dydv.\notag
            \end{eqnarray}
            By Assumption~A5, we write $$v_1(n)=O\left(\frac{1}{n^2h_th_y}\right).$$
            \noindent
            Likewise, from~\eqref{eq:thm:pw_conv:S4:v2n}--\eqref{eq:thm:pw_conv:S5:v3n}, by~\eqref{eq:model} and Assumptions~A1*, A3, A4 and~A5, we have
            \begin{eqnarray*}
                &&\hskip -20pt\lim_{n\rightarrow\infty}n_2v_{2}(n)\\
                &=&\int_{-\infty}^\infty\int_{-\infty}^\infty f_1^2(g(y))f_2(y) f_{\epsilon_1}^2(v-m(g(y))+m(g_0(y)))f_{\epsilon_2}(v)dvdy -N^2(g),\\
                &&\hskip -20pt\lim_{n\rightarrow\infty}n_1v_3(n)\\
                &=&\int_{-\infty}^\infty\int_{-\infty}^\infty\frac{f_2^2(y)}{g'(y)}f_1(g(y)) f_{\epsilon_1}(v-m(g(y))+m(g_0(y)))f_{\epsilon_2}^2(v) dvdy - N^2(g),
            \end{eqnarray*}
            and by Assumption~A5 for $g\in\G$, we obtain $$v_2(n)=O\left(\frac{1}{n}\right),\quad v_3(n)=O\left(\frac{1}{n}\right).$$

            \medskip
            Write $var(\Dn{g})$ as $v_1^\ast(n)+v_2^\ast(n)+v_3^\ast(n)+v_4^\ast(n)$, where
            \begin{eqnarray}
                v_1^\ast(n)&=&\frac{1}{(n_1n_2h_t)^2}\sum_{i=1}^{n_1}\sum_{j=1}^{n_2}var\left\{K_1\left(\frac {t_i-g(s_j)}{h_t}\right)\right\},\label{eq:thm:pw_conv:S6:v1n}\\
                v_2^\ast(n)&=&\frac{1}{(n_1n_2h_t)^2}\!\!\sum_{i=1}^{n_1}\sum_{j=1}^{n_2}\sum^{n_1}_{\substack {i'=1\\ (\neq i)}} \!\!cov\left\{\!\!K_1\!\!\left(\frac{t_i-g(s_j)}{h_t}\right),K_1\!\!\left(\frac{t_{i'}-g(s_j)} {h_t}\right)\right\},\label{eq:thm:pw_conv:S7:v2n}\\
                v_3^\ast(n)&=&\frac{1}{(n_1n_2h_t)^2}\!\!\sum_{i=1}^{n_1}\sum_{j=1}^{n_2}\sum^{n_2}_{\substack {j'=1\\ (\neq j)}} cov\left\{\!\!K_1\!\!\left(\frac{t_i-g(s_j)}{h_t}\right),K_1\!\!\left(\frac{t_i-g(s_{j'})}{h_t}\right)\right\},\label{eq:thm:pw_conv:S8:v3n}\\
                v_4^\ast(n)&=&\frac{1}{(n_1n_2h_t)^2}\sum_{i=1}^{n_1}\sum_{j=1}^{n_2}\sum_{\substack{i'=1\\(\neq i)}}^{n_1}\sum_{\substack{j'=1\\(\neq j)}}^{n_2}
                cov\left\{K_1\left(\frac{t_i-g(s_j)}{h_t}\right),K_1\left(\frac{t_{i'}-g(s_{j'})}{h_t}\right)\right\}.\notag
            \end{eqnarray}
            From the model specifications, $v_4^\ast(n)=0$. We now show that all the other $v_i^\ast(n)$'s tend to zero as $n\rightarrow\infty$.

            From~\eqref{eq:thm:pw_conv:S6:v1n}, by \eqref{eq:model} and Assumption~A3, we have
            $$v_1^\ast(n)=\frac{1}{n_1n_2h_t}v_{11}^\ast(n)-\frac{1}{n_1n_2}E^2(D_n(g)),$$ where
            \begin{eqnarray}
                v_{11}^\ast(n)=\int_{-\infty}^\infty\int_{-\infty}^\infty I_{\left(\frac{-g(y)}{h_t},\infty\right)}(z)K_1^2(z)f_1(g(y)+h_tz)f_2(y)dzdy.\notag
            \end{eqnarray}
            Arguing in a similar manner as in the case of $var(v_1(n))$ above, by Assumptions~A3, A4 and~A5, we have
            \begin{eqnarray}
                \lim_{n\rightarrow\infty}n_1n_2h_tv_{1}^\ast(n)=\int_{-\infty}^\infty K_1^2(z)dz\int_{-\infty}^\infty f_1(g(y))f_2(y)dy\notag
            \end{eqnarray}
            and consequently, by Assumption~A5, $$v_1^\ast(n)=O\left(\frac{1}{n^2h_t}\right).$$

            Likewise, by~\eqref{eq:model} and Assumptions~A3, A4 and~A5, we obtain from~\eqref{eq:thm:pw_conv:S7:v2n}--\eqref{eq:thm:pw_conv:S8:v3n},
            \begin{eqnarray}
                \lim_{n\rightarrow\infty}n_2v_{2}^\ast(n)=\int_{-\infty}^\infty f_1^2(g(y)) f_2(y)dy-D^2(g),\notag\\
                \lim_{n\rightarrow\infty}n_1v_{3}^\ast(n)=\int_{-\infty}^\infty\frac{f_1(g(y))f_2^2(y)}{g'(y)}dy-D^2(g)\notag
            \end{eqnarray}
            and by Assumption~A5 for $g\in\G$, $$v_2^\ast(n)=O\left(\frac{1}{n}\right),\quad v_3^\ast(n)=O\left(\frac{1}{n}\right).$$
            Thus, we have shown that all the $v_i^\ast(n)$'s tend to zero as $n\rightarrow\infty$.

            The theorem then follows by the continuous mapping theorem of convergence in probability.

\bigskip
            {\sc Proof of Theorem \ref{thm:uniqueness}.}
            Observe from \eqref{eq:L} and Assumption~A3* that
            \begin{eqnarray*}
                L(g)=\frac{\intinfinf f_1(g(y))f_2(y)f_{\epsilon_1+\epsilon_2}(m(g(y))-m(g_0(y))) dy}{\intinfinf f_1(g(y))f_2(y)dy}
            \end{eqnarray*}
            where $f_{\epsilon_1+\epsilon_2}$ is the convolution of the densities $f_{\epsilon_1}$ and $f_{\epsilon_2}$.
            In particular, $L(g_0)=f_{\epsilon_1+\epsilon_2}(0)$. Thus,
            \begin{eqnarray*}
                L(g_0)-L(g)=\frac{\intinfinf f_1(g(x))f_2(x)\left[f_{\epsilon_1+\epsilon_2}(0)-f_{\epsilon_1+\epsilon_2}(m(g(x))-m(g_0(x)))\right] dx}{\intinfinf f_1(g(x))f_2(x)dy}.
            \end{eqnarray*}
            By Lemma~\ref{lem:unimodality}, stated and proved below, $[f_{\epsilon_1+\epsilon_2}(0)-f_{\epsilon_1+\epsilon_2}(m(g(x))-m(g_0(x)))]\ge0$ for all $x$, which
            proves part~(a).

            In order that the last expression for $L(g_0)-L(g)$ happens to be zero for some $g\in\G$, the above difference must be equal to zero for all $x$ such that $f_1(g(x))f_2(x)>0$, i.e., for $x\in S_g\cap S_{g_0}$, where $S_g$ is as defined at the beginning of Section~\ref{sec:con}. Since $f_{\epsilon_1+\epsilon_2}$ is strictly unimodal at 0, this requirement reduces to $m(g(x))=m(g_0(x))$ for all $x\in S_g\cap S_{g_0}.$ It follows from Theorem~\ref{thm:identifiability} that $S_g=S_{g_0}$ and $g(x)= g_0(x)$ for all $x\in S_{g_0}$, which establishes part~(b).

        \begin{lemma}\label{lem:unimodality}
            Let $f_{\epsilon_1+\epsilon_2}$ be the convolution of the densities $f_{\epsilon_1}$ and $f_{\epsilon_2}$. Then under Assumption~A3*, $f_{\epsilon_1+\epsilon_2}$ is strictly unimodal at zero.
        \end{lemma}

        \begin{proof}
            Fix $0<u_1<u_2$. By Assumption~A3*, we observe that
            \begin{eqnarray}\label{eq:lem:unimodality:1}
                &&f_{\epsilon_1+\epsilon_2}(u_2)-f_{\epsilon_1+\epsilon_2}(u_1)\notag\\
                &=&\int_{-\infty}^{\infty}f_{\epsilon_1}(v)\{f_{\epsilon_2}(u_2-v)-f_{\epsilon_2}(v-u_1)\}dv\notag\\
                &=&\int_{-\infty}^{\infty}f_{\epsilon_1}(u_1+v)\{f_{\epsilon_2}(u_2-u_1-v)-f_{\epsilon_2}(v)\}dv\notag\\
                &=&\int_{-\infty}^{\frac{u_2-u_1}{2}}f_{\epsilon_1}(u_1+v)\{f_{\epsilon_2}(u_2-u_1-v)-f_{\epsilon_2}(v)\}dv\notag\\ &&\qquad\qquad+\int_{\frac{u_2-u_1}{2}}^{\infty}f_{\epsilon_1}(u_1+v)\{f_{\epsilon_2}(u_2-u_1-v)-f_{\epsilon_2}(v)\}dv\notag\\
                &=&\int_{\frac{u_2-u_1}{2}}^{\infty}f_{\epsilon_1}(u_2-v)\{f_{\epsilon_2}(v)-f_{\epsilon_2}(u_2-u_1-v)\}dv\notag\\ &&\qquad\qquad+\int_{\frac{u_2-u_1}{2}}^{\infty}f_{\epsilon_1}(u_1+v)\{f_{\epsilon_2}(u_2-u_1-v)-f_{\epsilon_2}(v)\}dv\notag\\
                &=&\int_{\frac{u_2-u_1}{2}}^{\infty}\{f_{\epsilon_1}(u_2-v)-f_1(u_1+v)\}\{f_{\epsilon_2}(v)-f_{\epsilon_2}(u_2-u_1-v)\}dv\notag\\
                &=&\int_0^{\infty}\left\{f_{\epsilon_1}\left(\frac{u_2+u_1}{2}-v\right)-f_{\epsilon_1}\left(\frac{u_2+u_1}{2}+ v\right)\right\}\notag\\
                &&\qquad\qquad\times\left\{f_{\epsilon_2}\left(\frac{u_2-u_1}{2}+v\right)-f_{\epsilon_2}\left(\frac{u_2-u_1}{2}-v\right)\right\}dv\notag\\
                &<&0.\notag
            \end{eqnarray}
            By a similar argument, the same inequality holds for $u_2<u_1<0$.
        \end{proof}

        Next, we state and prove a lemma which is used to prove uniform convergence of the functional $L_n$.
        \begin{lemma}\label{lem:Ln_oscillation}
            Let the class $\G$ be as described at the beginning in this section. Then under Assumptions~A1*, A3, A4 and A5,
            \begin{enumerate}
                \item[(a)] for $g,\tilde{g}\in\G$, $$|L_n(\tilde{g})-L_n(g)|\le B_n(\tilde{g})\|g-\tilde{g}\|,$$ where
                \begin{eqnarray}
                    B_n(\tilde{g})&=&c^{-2}M_K^{'}\{\Nn{\tilde{g}}+U_n\Dn{\tilde{g}}\},\label{eq:Bn}\\ U_n&=&(n_1n_2h_y)^{-1}\sum_{i=1}^{n_1}\sum_{j=1}^{n_2}K_2\{(y_{1i}-y_{2j})/h_y\},\notag
                \end{eqnarray}
                the functionals $L_n$, $N_n$ and $D_n$ are as defined in~\eqref{eq:Ln}, \eqref{eq:Nn} and \eqref{eq:Dn} respectively and $c$ and $M_K'$ be such that $0<c\le K_i(x)$ and $|K_i'(x)|\le M_K'\qquad (i=1,2)$,
                \item[(b)] $U_n$ tends to a constant $U$ in probability as $n\rightarrow\infty$, where $$U=\int_{-\infty}^{\infty}\int_{-\infty}^{\infty}\int_{-\infty}^{\infty} f_1(x)f_2(y)f_{\epsilon_1}(v-m(x)+m(g_0(y)))f_{\epsilon_2}(v)dxdydv,$$
                \item[(c)] $B_n(\tilde{g})$ tends to $B(\tilde{g})$ in probability as $n\rightarrow\infty$, where
                \begin{equation}
                    B(\tilde{g})=c^{-2}M_K^{'}\{N(\tilde{g})+U D(\tilde{g})\}\label{eq:B},
                \end{equation}
                the functionals $N$ and $D$ being as defined in~\eqref{eq:N} and \eqref{eq:D} respectively.
            \end{enumerate}
        \end{lemma}

        \begin{proof}
            From~\eqref{eq:Dn}, we have $D_n(g)\ge c/h_t$. Using this inequality, we have from~\eqref{eq:Ln_L}
            \begin{eqnarray}\label{eq:lem:Ln_oscillation:1}
               |L_n(\tilde{g})-L_n(g)|\leq \frac{h_t^2}{c^2}\left\{\Dn{\tilde{g}}\left|\Nn{\tilde{g}}-\Nn{g} \right|+\Nn{\tilde{g}}\left|\Dn{\tilde{g}}-\Dn{g}\right|\right\}.\quad\mbox{}
            \end{eqnarray}
            From~\eqref{eq:Nn}, we have
            \begin{eqnarray*}
                &&\hskip -30pt|\Nn{\tilde{g}}-\Nn{g}|\\
                &\leq&\frac{1}{n_1n_2h_th_y}\sum_{i=1}^{n_1}\sum_{j=1}^{n_2}\!K_{2}\!\left(\!\frac{y_{1i}-y_{2j}}{h_y}\!\right)\!\left| \!K_{1}\!\left(\!\frac{t_i-\tilde{g}(s_j)}{h_t}\!\right)-\!K_{1}\!\left(\!\frac{t_i-g(s_j)}{h_t}\!\right)\right|,
            \end{eqnarray*}
            which, by the mean value theorem, reduces to
            \begin{eqnarray}\label{eq:Lem:Ln_oscillation:2}
                |\Nn{\tilde{g}}-\Nn{g}|\leq\frac{M_K^{'}}{h_t^2}\|g-\tilde{g}\|U_n.
            \end{eqnarray}

            Similarly, from the expressions of $D_n(g)$ in~\eqref{eq:Dn} and by the mean value theorem, we have
            \begin{eqnarray}\label{eq:Lem:Ln_oscillation:3}
                |\Dn{\tilde{g}}-\Dn{g}|\leq\frac{M_K^{'}}{h_t^2}\|g-\tilde{g}\|.
            \end{eqnarray}

            Part (a) follows by combining~\eqref{eq:Lem:Ln_oscillation:2} and \eqref{eq:Lem:Ln_oscillation:3} with~\eqref{eq:lem:Ln_oscillation:1}.

            \medskip
            From the expression of $U_n$ in this Lemma, by~\eqref{eq:model} and Assumptions~A3,
            \begin{eqnarray}
                E(U_n)&=&\intinfinf\intinfinf\intinfinf\intinfinf \frac{1}{h_y}K_2\left(\frac{m(x)-m(\g(y))+u-v}{h_y}\right)\notag\\
                &&\times\quad f_1(x)f_2(y)f_{\epsilon_1}(u)f_{\epsilon_2}(v) dxdydudv\notag\\
                &=&\intinfinf\intinfinf\intinfinf\intinfinf K_2(w) f_1(x)f_2(y)\notag\\
                &&\times\quad f_{\epsilon_1}(v-m(x)+m(g_0(y))+wh_y)f_{\epsilon_2}(v)dxdydwdv.\notag
            \end{eqnarray}
            By a similar argument as in the proof of Theorem~\ref{thm:pw_conv}, we obtain, by Assumptions~A1*, A3, A4 and~A5,
            \begin{eqnarray}\label{eq:lem:Ln_oscillation:4:Un_conv}
                \lim_{n\rightarrow\infty}E(U_n)=U.
            \end{eqnarray}

            Write $var(U_n)$ as $v_1(n)+v_2(n)+v_3(n)+v_4(n)$ where
            \begin{eqnarray}
                v_1(n)&=&\frac{1}{(n_1n_2h_y)^2}\SUM{i}{1}{n_1}\SUM{j}{1}{n_2}var\left\{K_2\left(\frac{ y_{1i}-y_{2j}}{h_y}\right)\right\},\notag
            \end{eqnarray}
            \begin{eqnarray}
                v_2(n)&=&\frac{1}{(n_1n_2h_y)^2}\SUM{i}{1}{n_1}\SUM{j}{1}{n_2}\nSUM{i'}{1}{n_1}{i} cov\left\{K_2\left(\frac{y_{1i}-y_{2j}}{h_y}\right),K_2\left(\frac{y_{1i'}-y_{2j}}{h_y}\right)\right\},\notag
            \end{eqnarray}
            \begin{eqnarray}
                v_3(n)&=&\frac{1}{(n_1n_2h_y)^2}\SUM{i}{1}{n_1}\SUM{j}{1}{n_2}\nSUM{j'}{1}{n_2}{j}
                cov\left\{K_2\left(\frac{y_{1i}-y_{2j}}{h_y}\right),K_2\left(\frac{y_{1i}-y_{2j'}}{h_y}\right)\right\},\notag
            \end{eqnarray}
            \begin{eqnarray}
                v_4(n)&=&\frac{1}{(n_1n_2h_y)^2}\SUM{i}{1}{n_1}\SUM{j}{1}{n_2}\nSUM{i'}{1}{n_1}{i}
                \nSUM{j'}{1}{n_2}{j}cov\left\{K_2\left(\frac{y_{1i}-y_{2j}}{h_y}\right),K_2\left(\frac{y_{1i'}-y_{2j'}}{h_y}\right)\right\}.\notag
            \end{eqnarray}
            From the specifications of the model, it follows that $v_4(n)=0$. We now show that all the other $v_i$'s also tend to zero as $n\rightarrow\infty$.

            By a similar argument as in the proof of Theorem~\ref{thm:pw_conv}, it follows that under Assumptions~A1*, A3, A4 and~A5,
            \begin{eqnarray*}
                &&\hskip -30pt\lim_{n\rightarrow\infty}\!\!n_1n_2h_yv_1(n)\!\!\\
                &=&\!\!\intinfinf\!\!\!\!\!\! K^2_2(w)dw\!\!\! \intinfinf\!\!\intinfinf\!\!\!\!\intinfinf\!\!\!\!\!\! f_1(x)f_2(y) f_{\epsilon_1}(v-m(x)+m(g_0(y)))f_{\epsilon_2}(v)dxdydv,\\
                &&\hskip -30pt\lim_{n\rightarrow\infty}\!\!n_2v_2(n)\!\!\\
                &=&\!\!\intinfinf\!\!\intinfinf\!\!f_2(y)f_{\epsilon_2}(v)\left\{\intinfinf\!\!\!\!\!\!f_1(x) f_{\epsilon_1}(v-m(x)+m(g_0(y)))
                dx\right\}^2\!\!\!dvdy\!-\!U^2,\\
                &&\hskip -30pt\lim_{n\rightarrow\infty}\!\!n_1v_3(n)\!\!\\
                &=&\!\!\intinfinf\!\!\intinfinf\!\!f_1(x)f_{\epsilon_1}(v)\left\{\intinfinf\!\!\!\!\!\! f_2(y) f_{ \epsilon_2} (v+m(x)-m(g_0(y)))dy\right\}^2\!\!\!dvdx\!-\!U^2.
            \end{eqnarray*}
            Therefore, by Assumption~A5 we have, $$v_1(n)=O\left(\frac{1}{n^2h_y}\right),\quad v_2(n)=O\left(\frac{1}{n}\right),\quad v_3(n)=O\left(\frac{1}{n}\right).$$ This completes the proof of part (b).

            \medskip
            Part (c) follows from~\eqref{eq:thm:pw_conv:S1:Nng_conv}, \eqref{eq:thm:pw_conv:S2:Dng_conv} in supplementary materials, and~\eqref{eq:lem:Ln_oscillation:4:Un_conv}.
        \end{proof}

            {\sc Proof of Theorem~\ref{thm:unif_conv}.}
            By triangular inequality for any $g, \tilde{g}\in\G_0$, we have
            \begin{eqnarray}\label{eq:thm:unif_conv_1}
                |L_n(g)-L(g)|\leq|L(\tilde{g})-L(g)|+|L_n(g)-L_n(\tilde{g})|+|L_n(\tilde{g})-L(\tilde{g})|.
            \end{eqnarray}

            Set $\epsilon>0$. By Lemma~\ref{lem:L_unif_cont}, stated and established below, there exists $\delta_{\epsilon}>0$ such that $|L(g)-L(\tilde{g})|<\epsilon/3$ when $\|g-\tilde{g}\|<\delta_{\epsilon}$.

            By Lemma~\ref{lem:Ln_oscillation}, stated and proved above, $|L_n(\tilde{g})-L_n(g)|\le B_n(\tilde{g})\|g-\tilde{g}\|$ where as $n\rightarrow\infty$, $B_n(\tilde{g})$ goes to $B(\tilde{g})$ in probability which implies
            \begin{eqnarray}\label{eq:thm:unif_conv_2}
                \lim_{n\rightarrow\infty}pr\left\{B_n(\tilde{g})>\max\left(\frac{\epsilon}{3\delta_{\epsilon}},2B(\tilde{g})\right)\right\}= 0
            \end{eqnarray}
            where the functionals $B_n$ and $B$ are as in~\eqref{eq:Bn} and~\eqref{eq:B} respectively. For a given $\tilde{g}\in\G_0$, let $\delta(\tilde{g},\epsilon)=\min\{(6B(\tilde{g}))^{-1}\epsilon,\delta_{\epsilon}\}$ if $B(\tilde{g})>0$ or $\delta_{\epsilon}$ if $B(\tilde{g})=0$. Clearly $\{\mathcal{N}_{\delta(\tilde{g},\epsilon)}(\tilde{g}):\tilde{g}\in\G_0\}$ is an open cover of $\G_0$ where $\mathcal{N}_{\eta}(\tilde{g})=\left\{g:\|g-\tilde{g}\|<\eta\right\}$. By the compactness of $\G_0$, this cover contains a finite sub-cover, say
            $\{\mathcal{N}_{\delta(\tilde{g}_j,\epsilon)}(\tilde{g}_j)\ j=1\ldots k_\epsilon\}$ for some finite $k_\epsilon$, of $\G_0$. Therefore, from~ \eqref{eq:thm:unif_conv_1},
            \begin{eqnarray}
               &&\hskip -30pt\sup_{g\in\G_0}|L_n(g)\!-\!L(g)|\notag\\
               &\leq&\!\!\max_{j=1,\ldots,k_\epsilon}\!\!\!\left\{\!\sup_{g\in\mathcal{N}_{\delta(\tilde{g}_j,\epsilon)}(\tilde{g}_j)} \!\!\!\!\!\!\!\!\!\!\!|L(\tilde{g}_j)\!-\!L(g)|+\!\!\!\!\!\!\!\!\!\!\sup_{g\in\mathcal{N}_{\delta(\tilde{g}_j,\epsilon)}(\tilde{g}_j)} \!\!\!\!\!\!\!\!\!\!\!|L_n(g)\!-\!L_n(\tilde{g}_j)|+|L_n(\tilde{g}_j)\!-\!L(\tilde {g}_j)|\right\}\notag\\
               &\leq&\frac{\epsilon}{3}+\max_{j=1,\ldots,k_\epsilon}\delta(\tilde{g}_j,\epsilon)B_n(\tilde{g}_j)+\sum_{j=1}^{k_\epsilon} |L_n(\tilde{g} _j)-L(\tilde{g}_j)|.\notag
            \end{eqnarray}
            Consequently,
            \begin{eqnarray}
                &&\hskip -30pt pr\!\left(\sup_{g\in\G_0}|L_n(g)\!-\!L(g)|\!>\!\epsilon\!\right)\notag\\ \!\!&\leq&\!\!\sum_{j=1}^{k_{\epsilon}}\!pr\!\left(B_n(\tilde{g}_j)>\frac{\epsilon}{3\delta(\tilde{g}_j,\epsilon)} \right)\!\!+\!pr\!\left(\sum_{j=1}^{k_\epsilon}\!|L_n(\tilde{g }_j)\!-\!L(\tilde{g}_j)|>\frac{\epsilon}{3}\right)\!\!.\notag
            \end{eqnarray}
            Each summand of the first term goes to zero by \eqref{eq:thm:unif_conv_2}, while those of the second term go to zero by Theorem~\ref{thm:pw_conv}.

        \begin{lemma}\label{lem:L_unif_cont}
            Let $\G_0$ be as defined in Theorem~\ref{thm:unif_conv}. Then under Assumptions~A1* and~A3, the functional $L$ in~\eqref{eq:L} is uniformly continuous on $\G_0$.
        \end{lemma}

        \begin{proof}
            Let us recall the representation~\eqref{eq:Ln_L} of $L(g)$ and the expressions of $N(g)$ and $D(g)$ in~\eqref{eq:N}~and~\eqref{eq:D} respectively. We first show that $N(g)$ is continuous on $\G_0$. Let $g\in\G_0$ and $\{g_k\}$ be a sequence of functions in $\G_0$ such that $\lim_{k\rightarrow\infty}\|g_k-g\|=0$. By Assumption~A3, it follows that the integrand in $N(g_k)$ is bounded by an integrable function $M_f^2f_2(y)f_{\epsilon_2}(v)$ where $M_f$ is the common upper bound of $f_1$, $f_2$, $f_{\epsilon_1}$ and $f_{\epsilon_2}$. Further, by the continuity of $f_1$ and $f_{\epsilon_1}$ (Assumption~A3) and $m$ (Assumption~A1*), and the fact that $g_k$ converges to $g$ pointwise, it follows that the integrand tends to $f_1(g(y))f_2(y)f_{\epsilon_1}(v-m(g(y))+m(g_0(y)))f_{\epsilon_2}(v)$ as $k\rightarrow\infty$. Therefore, by dominating convergence theorem $\lim_{k\rightarrow\infty}N(g_k)=N(g)$ i.e.\ $N$ is continuous on $\G_0$.

             Similar arguments show that $D(g)$ is continuous on $\G_0$ as well. Further, by Assumption~A3 and the condition that $g[c,d]\cap[a,b]$ we have that $D(g)>0$ for all $g\in\G_0$. This gives that $L(g)$ is continuous on $G_0$. The lemma is, therefore, followed from the compactness of $G_0$.
        \end{proof}

            {\sc Proof of Theorem~\ref{thm:Lngn_hat_conv}.}
            Fix $\epsilon > 0$. We have
            \begin{eqnarray}
                &&\hskip -30pt pr(|L_n(\hat{g}_n)-L(g_0)|>\epsilon)\notag\\
                &\leq& pr(|L_n(\hat{g}_n)-L(g_0)|>\epsilon,|L_n(g_0)-L(g_0)|\leq\epsilon,|L_n(\hat{g}_n)-L(\hat{g}_n)|\leq\epsilon)\notag\\
                &&\qquad+pr(|L_n(\hat{g}_n)-L(\hat{g}_n)|>\epsilon)
                +pr(|L_n(g_0)-L(g_0)|>\epsilon),\label{eq:thm:Lngn_hat_conv:1}
            \end{eqnarray}
            where $\hat{g}_n$ is as in~\eqref{eq:gn_hat}. We show that all the three terms on the right side of (\ref{eq:thm:Lngn_hat_conv:1}) are arbitrarily small.

            From~\eqref{eq:gn_hat}, we have $L_n(g_0)\le L_n(\gn)$. Thus, if $|L_n(g_0)-L(g_0)|\leq\epsilon$, $L_n(\hat{g}_n)\ge L(g_0)-\epsilon$.
            By Theorem~\ref{thm:uniqueness}(a), if $|L_n(\hat{g}_n)-L(\hat{g}_n)|\leq\epsilon$ then $L_n(\hat{g}_n)\le L(g_0)+\epsilon$.
            Therefore, if $|L_n(g_0)-L(g_0)|\leq\epsilon$ and $|L_n(\hat{g}_n)-L(\hat{g}_n)|\leq\epsilon$ then $|L_n(\hat{g}_n)-L(g_0)|\leq\epsilon$, which makes the first term on the right side of~\eqref{eq:thm:Lngn_hat_conv:1} zero. Observing that $|L_n(\gn)-L(\gn)|\leq\sup_{g \in\G_0}|L_n(g)-L(g)|$, by Theorem~\ref{thm:unif_conv} the second term tends to zero while by Theorem~\ref{thm:pw_conv} the last term goes to zero as $n\rightarrow\infty$.

\bigskip
            {\sc Proof of Theorem~\ref{thm:consistency}.}
            Suppose that $\hat{g}_n$ does not tend to $g_0$ in probability as $n\rightarrow\infty$, i.e.\ there exists an $\epsilon>0$ and a $\delta
            >0$ such that $P\{\|\hat{g}_n-g_0\|\geq\epsilon\}>\delta$ infinitely often. Consider a closed subset $\mathcal{N}_\epsilon^c(g_0)=\{g:\|g-g_0\|\ge\epsilon,\;g\in\G_0\}$. By Lemma~\ref{lem:L_unif_cont} and the compactness of~$\G_0$, there exists a
            $\tilde{g}\in \mathcal{N}_\epsilon^c(g_0)$ such that $\tilde{g}=\argmax_{g\in\mathcal{N}_\epsilon^c(g_0)}L(g)$. Denote $\eta=L(g_0)-L(\tilde{g})$. By Theorem~\ref{thm:uniqueness}(b), $\eta>0$. We have $|L_n(\hat{g}_n)-L(g_0)|\ge |L(g_0)-L(\hat{g}_n)|-|L_n(\hat{g}_n)-L(\hat{g}_n)|$. Thus, if $\hat{g}_n\in\mathcal{N}_\epsilon^c(g_0)$ and $\sup_{g\in\G_0}|L_n(g)-L(g)|<\eta/2$ then $|L_n(\hat{g}_n)-L(g_0)|>\eta/2$. Consequently,
            \begin{eqnarray}
                &&pr\left\{|L_n(\hat{g}_n)-L(g_0)|>\frac{\eta}{2}\right\}\notag\\
                &\geq&pr\{\|\hat{g}_n-g_0\|\geq\epsilon\}+pr\bigg\{\sup_{g\in\G_0}|L_n(g)-L(g)|<\frac{\eta}{2}\bigg\}-1\notag.
            \end{eqnarray}
            The first term on the right hand side is greater than $\delta$ infinitely often while by Theorem~\ref{thm:unif_conv}, the second term is greater than $1-\frac{\delta}{2}$ for all but finitely many $n$. Therefore, $pr(|L_n(\hat{g}_n)-L(g_0)|>\eta/2)>\delta/2\ \mbox{infinitely often}$, which contradicts Theorem~\ref{thm:Lngn_hat_conv}.
    \bibliographystyle{rss}
    \bibliography{mybib}
\end{document}